\documentclass[11pt,reqno]{amsart}

\usepackage{amsmath,amsfonts,amssymb,commath}
\usepackage{graphicx,harvard}

\usepackage{multibib}

\usepackage{booktabs}
\usepackage[load-configurations=version-1]{siunitx} 
\usepackage{siunitx}
\usepackage{multirow}
\usepackage{pdfpages}

\usepackage{caption2}
\usepackage{setspace}

\newcommand{\Expect}{\mathbb E}

\newcommand{\Prob}{\mathbb P}

\newtheorem{theorem}{Theorem}
[section]

\newtheorem{proposition}[theorem]{Proposition}

\newtheorem{definition}[theorem]{Definition}
\newtheorem{lemma}[theorem]{Lemma}

\newtheorem{assumption}[theorem]{Assumption}

\newtheorem{example}[theorem]{Example}

\def\comp{{\rm comp}}
\def\taubar{{\overline\tau}}
\def\bbar{{\overline b}}

\def\bigO{{\sf O}}

\def\pubar{{\overline p}}
\def\plbar{{\underline p}}
\def\vlbar{{\underline v}}
\def\vubar{{\overline v}}    \def\Tbar{{\overline T}}

\begin{document}
\title[Learning Underspecified Models]{\bf Learning Underspecified Models}

\date{\today}

\author{In-Koo Cho}
\address{Department of Economics, Emory University, Atlanta, GA 30322 USA}
\email{icho30@emory.edu}
\urladdr{https://sites.google.com/site/inkoocho}

\author{Jonathan Libgober}
\address{Department of Economics, University of Southern California, Los Angeles, CA 90089 USA}
\email{libgober@usc.edu}
\urladdr{http://www.jonlib.com/}

\thanks{The first author gratefully acknowledges financial support from the National Science Foundation (SES-1952882) and the Korea Research Foundation (2018S1A5A2A01029483).}

\begin{abstract}
  This paper examines whether one can learn to play an optimal action while only knowing part of true specification of the environment.   We choose the optimal pricing problem as our laboratory, where the monopolist is endowed with an   \emph{underspecified} model of the market demand, but can observe market outcomes.   In contrast to conventional learning models where the model specification is complete and exogenously fixed, the monopolist has to learn the specification and the parameters of the demand curve from the data.   We formulate the learning dynamics as an algorithm that forecast the optimal price based on the data, following the machine learning literature
(\citeasnoun{Shalev-ShwartzandBen-David14}).
Inspired by PAC learnability, we develop a new notion of learnability by requiring that the algorithm must produce an accurate forecast with a reasonable amount of data {\em uniformly} over the class of models consistent with the part of the true specification.  In addition, we assume that the monopolist has a lexicographic preference over the payoff and the complexity cost of the algorithm, seeking an algorithm with a minimum number of parameters subject to PAC-guaranteeing the optimal solution (\citeasnoun{Rubinstein86}).  We show that for the set of demand curves with strictly decreasing uniformly Lipschitz continuous marginal revenue curve, the optimal algorithm recursively estimates the slope and the intercept of the linear demand curve, even if the actual demand curve is not linear.  The monopolist chooses a misspecified model to save computational cost, while learning the true optimal decision uniformly over the set of underspecified demand curves.\\[2pt] 
\noindent {\sc Keywords.} Optimal Sales Mechanism, Underspecified model, PAC guarantee, Complexity Cost, Misspecification
\end{abstract}

\maketitle

\section{Introduction}

A rational decision maker is endowed with a complete and correctly specified model of the state and the relationship between his action and consequence (\citeasnoun{Simon87}).   Yet, it is an excessively strong assumption that a decision maker is fully aware of {\em all relevant states} and knows {\em precise specification} of the function that maps his strategy to payoff.  This paper supposes instead that the decision maker knows only a part of true specification of the environment.  Thus, the decision maker is endowed with an {\em underspecified} model,  facing model uncertainty (\citeasnoun{HansenandSargent07}), as multiple models including the true model are consistent with the partial specification.\footnote{An underspecified decision problem differs from a misspecified model.   A misspecified model completes the specification of the decision problem but the specification of the model differs from the actual specification.}    Our research objective is to understand how a decision maker learns to choose a specification to find the optimal decision under the model uncertainty.

We choose the monopolistic profit maximization problem (\citeasnoun{Myerson81}) as the laboratory, which has been extensively studied in economics and computer science.  Our exercise has four features, which together differentiate our research from the existing studies.   First, in a sharp contrast to the Bayesian models, the monopolist lacks any information to form a prior over the set of possible missing variables.  The seller is facing model uncertainty (\citeasnoun{HansenandSargent07}).

Second, the robust approach to the model uncertainty typically allows no opportunity for the decision maker to learn from the data  (\citeasnoun{HansenandSargent07} and \citeasnoun{BergemannandMorris05b}).\footnote{A remarkable exception is  a series of papers by Hansen and Sargent (e.g., \citeasnoun{HansenandSargentJET2020}).}
In our case,  the seller has an opportunity to observe data to learn about the demand curve.   In conventional learning models, the specification of the model is exogenously fixed and the algorithm calculates the unknown parameters.    In our exercise, the seller has to choose both the specification and the parameters of the demand curve to learn from the data.

Third, the robust approach focuses on the best response against the worst possible conjecture so that the strategy can perform ``reasonably well'' over the admissible set of demand curves, if not optimally.  In our case, the learning algorithm must produce a good approximation to the true optimal strategy with high confidence {\em uniformly} over the set of admissible demand curves.   If so, we say that the algorithm is uniformly learnable.  
Our notion of learnability inspired by  the PAC criterion in the machine learning literature (\citeasnoun{Shalev-ShwartzandBen-David14}).
Given $\epsilon>0$ error bound and $1-\rho>0$ confidence requirement, we search for a learning algorithm ${\mathcal A}$ and a stopping time $T(\epsilon,\rho)$ so that after $T(\epsilon,\rho)$ time steps, algorithm ${\mathcal A}$ produces a forecast  that is within $\epsilon$ neighborhood of the true optimal strategy with probability $1-\rho$ uniformly over the set of admissible demand curves. PAC criterion also disciplines the data complexity. The number $T(\epsilon,\rho)$ of the time steps should increase at the rate of a polynomial function of $1/\epsilon$ and at the logarithmic rate of $1/\rho$.

Fourth, the  monopolist in our exercise bears the computational cost.  We consider two sources of computational cost: the source of data and the complexity of an algorithm.
  The computer science literature (e.g., \citeasnoun{ColeRoughgarden2014}) focuses on the direct revelation game, where the buyers report their private valuation.   The private information of the buyer is freely available to the algorithm.    While the direct mechanism is a widely used mathematical model, it is rare that an actual trading protocol is a revelation game.  Instead, one has to invert the strategy of a buyer to infer his valuation, which is not a trivial exercise.  We assume that public data is cheaper than reported private information.    We measure the complexity of an algorithm by the number of parameters an algorithm has to estimate.\footnote{This type of complexity measure is consistent with the complexity measure for a neural work (cf. \citeasnoun{Rumelhartetal86}, \citeasnoun{Wasserman89} and \citeasnoun{Weisbuch90}). }   Following \citeasnoun{Rubinstein86}, we assume that the monopolist has a lexicographic preference over the payoff and the complexity cost.    The seller first seeks a PAC guaranteeing algorithm to maximize profit and then, chooses the simplest learning algorithm among those that achieve the maximum payoff.

We show that a uniformly learnable algorithm must estimate at least two parameters (Proposition \ref{pr: lower bound}).   We construct an algorithm estimating two parameters that learns uniformly the optimal strategy of the monopolist over the set of demand curves with strictly decreasing uniformly Lipschitz continuous marginal revenue curve.  The algorithm recursively estimates a linear demand curve (which can be parameterized by the slope and the intercept) using price and quantity, learning uniformly the optimal strategy of the monopolist (Theorem \ref{thm:mainresult}).  

The linear demand curves may be a grossly misspecified model if the actual demand curve is highly non-linear.   Nevertheless, within a ``reasonable amount'' of time, the monopolistic seller behaves as if he knows the actual demand curve.    While a monopolist could use an algorithm calculating  parameters of a correctly specified model, such as the non-parametric estimation algorithm (\citeasnoun{ColeRoughgarden2014}), the monopolist chooses a simpler model of demand to save computational costs. In this sense, a misspecified model would be a representation of the procedural rationality of the seller (\citeasnoun{OsborneandRubinstein98}).

The rest of the paper is organized as follows. Section \ref{Literature Review} reviews the literature, clarifying our contribution beyond the existing literature. In Section \ref{Description}, we formally describe the problem and define the basic concepts, with a summary of the main result, where we calculate the lower bound of the complexity of algorithms to calculate the optimal price.   In Section \ref{Construction}, we construct the algorithm, which is the simplest algorithm among algorithms that learns uniformly the optimal strategy. We state the main result, whose proof is in the appendix. Instead of a formal proof, Section \ref{Heuristics} informally illustrates how the monopolist can efficiently and uniformly learn the optimal price of a non-linear demand curve through a linear demand curve.   Section \ref{Numerical Experiments} reports numerical exercises to show that the performance of our algorithm is comparable to a more elaborate algorithm of \citeasnoun{ColeRoughgarden2014}.
We show how uniform learnability is related to the notion of PAC guarantee in computer science.  Section \ref{Concluding Remarks} concludes the paper.

\section{Literature Review}
\label{Literature Review}

Our paper differs from three main approaches to investigating the decision problem with underspecified models.     The most prominent approach is that a decision-maker chooses his actions assuming uncertainty about the environment resolves in favor of the worst-case (e.g., \citeasnoun{HansenandSargent07}, \citeasnoun{Carroll2015}, \citeasnoun{Carroll2017}, \citeasnoun{Du2018} and \citeasnoun{LibgoberMu2021}). The choice that maximizes the objective under the worst conjecture will typically differ from the optimal solution against the truth, sometimes significantly so. While data could, in principle, bring the decision-maker closer to optimality, typically, these models are silent on how the decision-maker might do this.\footnote{A remarkable exception is a series of papers by Hansen and Sargent (e.g., \citeasnoun{HansenandSargentJET2020}).}   In contrast, the monopolist in our model seeks to find a \emph{true} optimal price for an unknown demand curve by observing data.

In the second approach, the modeler completes the specification of the decision problem by imposing a (parametric) model, where the decision maker estimates the parameters using data. While this approach does allow the decision-maker to learn from data, the specification is fixed exogenously, and often  excludes the true mapping from actions into outcomes (e.g., \citeasnoun{ChoandKasa15}, \citeasnoun{ChoandKasa17}, \citeasnoun{HKS2018}, \citeasnoun{EPY2021}, \citeasnoun{FII2021a} and \citeasnoun{FudenbergLanzaniStrack}).\footnote{In the operations research, \citeasnoun{BesbesandZeeviMS2015} examined the monopoly profit maximization problem, where the monopolist is endowed with a linear demand curve, estimating the demand curve to choose the price. \citeasnoun{BesbesandZeeviMS2015} demonstrated that one could construct an algorithm to let the monopolist learn the actual optimal price for a class of demand curves satisfying a set of regularity conditions.}
The monopolist in our model, by contrast, is not committed to a particular specification.  The monopolist can use the non-parametric estimation method of the demand curve to avoid the misspecification problem.   The monopolist chooses the linear model with two parameters to minimize the computational cost, PAC guaranteeing the optimal price. The linear model, which is misspecified, is not imposed by the modeler but derived from optimization by the monopolist.

This paper follows a third approach initiated by machine learning models\footnote{\citeasnoun{Nisanetal2007} report the early applications of the algorithms to game theoretic models, including the monopoly problem.} (e.g., \citeasnoun{HuangMansourRoughgarden}, \citeasnoun{ColeRoughgarden2014}, \citeasnoun{GonczWeinberg2018} and \citeasnoun{GoncalvesFurtado}). A typical approach in this literature (\citeasnoun{ColeRoughgarden2014}) assumes that the seller non-parametrically estimates the demand curve. This approach avoids the underspecification issue, and the seller can achieve the true optimal curve as the estimator converges to an actual demand curve. The investigation then focuses on the consistency of the estimator and data complexity.

Following \citeasnoun{ColeRoughgarden2014}, we freely borrow the critical concepts developed in the machine learning literature (e.g., \citeasnoun{Shalev-ShwartzandBen-David14}). We depart from \citeasnoun{ColeRoughgarden2014}, however, by considering the complexity of the algorithm and the cost of obtaining data. An essential advantage of the non-parametric estimation technique is to avoid the misspecification of the demand curve. The downside is that the estimation algorithm is complex because the estimator requires a possibly unbounded number of parameters to represent a highly non-linear demand curve.   If the monopolistic seller incurs the cost of storing the estimator in the memory, he will search for an algorithm that uses a simpler specification of the demand curve.

\section{Description}
\label{Description}

\subsection{Demand} 

There are $N$ buyers, each of whom is indexed by $i\in\{1,\ldots,N\}$ and is endowed with reservation value $v_i\in [\vlbar,\vubar]$.  Let $F_i(v_i)$ be the distribution of valuation of buyer $i$. We assume that $v_i$ and $v_j$ are independent $\forall i\ne j$.   Given $p$, buyer $i$ purchases one unit of the good if $p\le v_i$.  If the seller charges $p$, the (normalized) aggregate demand is
\[
q= \frac{1}{N}\sum_{i=1}^N {\mathbb I}( v_i\ge p)
\ \ \text{ and } \ \
\Expect q =1- \frac{1}{N}\sum_{i=1}^NF_i(p).
\]
 Define
\[
F(p)=\frac{1}{N}\sum_{i=1}^NF_i(p)
\ \ \text{ and } \ \ 
\epsilon_2=q -(1-F(p)), 
\]
where $\Expect\epsilon_2=0$.

We interpret $1-F(p)$ as the expected quantity of sales. The actual amount $q$ of sales as a random variable, whose expected value is $1-F(p)$ if the price is $p$.   If $p=0$, $q=1$ with probability 1, and $F(p)\rightarrow 1$ as $p\rightarrow\vubar$.   We can treat $F$ as a distribution function.   Let $f$ be the density function of $F$.

\subsection{What the Monopolist Knows}

At the beginning, the complete specification $F$ is not available to the monopolist.    The monopolist does not have a prior over the set of feasible demand curves.
Instead, the monopolist knows a specific set of actual properties of $F$, which is common knowledge among all players.   We write the set of distributions that satisfy the said properties as ${\mathcal F}$.

Let us list the properties we are interested in and define the distributions accordingly.

\smallskip
\begin{itemize}
\item[{\sf IH}]
The support of $F$ is $[\plbar,\pubar]$ and $f(p)>0$ $\forall p\in [\plbar,\pubar]$.
$\forall F\in {\mathcal F}$, its density function $f$ is continuous over $[{\underline p}, {\overline p}]$ and its hazard rate
\[
\frac{f(p)}{1-F(p)}
\]
is increasing.
\end{itemize}

\smallskip
Let ${\mathcal F}^0$ be the set of all distributions over the buyer's valuations  satisfying the increasing hazard rate property.    In order to avoid the technical problems, we assume the Lipschitz continuity of the density function.  Define
  \[
    {\mathcal F}^\eta=\left\{ F\in {\mathcal F}^0 \mid \exists \eta>0, \ \forall p,p'\in[\plbar,\pubar], \ \left| f(p)-f(p')\right| <\eta | p-p |
      \right\}
  \]
as the collection of all feasible distributions over the buyer's valuations.
${\mathcal F}^\eta$ is a compact convex subset of ${\mathcal F}^0$, and is a large subset of ${\mathcal F}^0$ if we choose large $\eta>0$.     We state the useful properties in a lemma without proof.

\begin{lemma}  \label{lm: useful uniform}
  \begin{enumerate}
  \item $\forall\eta>0$, ${\mathcal F}^\eta$ is (sequentially) compact.
  \item $\forall\eta>0$, there exists a compact set $K$ in the interior of ${\mathbb R}^2_+$ so that $\forall F\in {\mathcal F}^\eta$, $(b^*(F),1-F(b^*(F)))\in K$.
  \item $\cup_{\eta>0} {\mathcal F}^\eta$ is a dense subset of ${\mathcal F}^0$.
\end{enumerate}
\end{lemma}

\subsection{The Seller's Objective}

In the equilibrium analysis, the monopolist knows the distribution $F$ of the buyer's valuations.  The seller does not have the prior probability over ${\mathcal F}^\eta$.
Instead, let us assume that the monopolist only knows that the actual distribution is an element of ${\mathcal F}^\eta$: $F\in {\mathcal F}^\eta$.

To incorporate the learning process by the monopolist, we consider a dynamic version of the static model, where the long run monopolist seller is facing a sequence of short run consumers with IID draws of their valuations.  

Time is discrete: $t=1,2,3,\ldots$. In each period, the monopolist post price $p_t$ and $N$ consumers enters the market, each of whom is endowed with reservation value $v_{i,t}$ drawn according to distribution function $F_i(\cdot)$.   Upon entering the market in period $t$, buyer $i$ purchases a single unit of the good paying $p_t$ if $v_{i,t}> p_t$ and leaves the market forever.   Since a buyer in period $t$ plays only once, his dominant strategy is to buy if $v_{i,t}> p_t$.  Define $q_t$ as the number of buyers whose valuation in period $t$ is more than $p_t$, which is the aggregate demand in period $t$.

In each period, the seller observes data $D_t$ which may include  (normalized) quantity $q_t$ and price $p_t$ charged in period $t$.   We differentiate the outcome from the data.   The outcome specifies what a player can observe.   By the data, we mean a subset of the outcome that the decision-maker uses as an input for the decision making process (or algorithm).    If processing information is costly, the decision-maker can choose to ignore some outcomes.   By definition, data is a subset of the outcome.    Let $\dim D_t$ be the number of components in $D_t$, which cannot be larger than the number of components of observed outcomes in period $t$.   The configuration of data is a part of the strategic choice of the monopolist in designing an algorithm.

\subsection{Algorithms}

While the monopolist is lacking a precise information about the actual distribution $F$, the monopolist can generate data $D_t$ in period $t$ to learn $F\in {\mathcal F}^\eta$.    Let ${\mathcal D}_t=(D_1,\ldots,D_{t-1})$ be a history at the beginning of period $t$, and ${\mathcal D}$ be the set of all histories.

Let $\Theta$ be the set of parameters that ``models'' $F\in {\mathcal F}$.
$\theta\in\Theta$ is a forecast (or a model) of the underlying aggregate demand curve.    We use forecast and model interchangeably throughout this paper.   Let $K$ be the number of parameters to model $F\in {\mathcal F}$.

\begin{assumption}
$\Theta$ is a compact subset of $\cup_{K\ge 1}{\mathbb R}^K$.
\end{assumption}

$\Theta$ is closed and bounded. By assuming that $\Theta$ is bounded, we assume that the decision-maker has some knowledge about the underlying demand curve.   If some component of $\theta$ is enormous, he knows the forecast makes no sense and discards the forecast.   On the other hand, it is a technical assumption that $\Theta$ is closed.    We can interpret $K$ as the number of parameters for a model.   We impose no upper bound for the number of parameters, admitting the non-parametric estimator as a feasible model.

We admit that the model of the demand can be misspecified.   If ${\mathcal F}$ is a collection of non-linear functions, $\theta$ may have infinitely many components.   If $\Theta\subset {\mathbb R}^K$ for finite $K$, then the model could be misspecified. As the number of the parameters increases, the model can better approximate the actual demand function.      We call ${\mathcal A}({\mathcal D}_t)$  the forecast conditioned on ${\mathcal D}_t$, or simply, the forecast.

To save the cost of decision making, the seller delegates his decision to an algorithm
\[
  {\mathcal A}: {\mathcal D}\rightarrow \Theta.
\]
The definition of an algorithm is very general.   A recursive estimation of parameters is an algorithm, as the statistical procedure takes the data as an input, and produces estimates of the variables of interest.   The number of components in ${\mathcal A}({\mathcal D}_t)$ or $D_t$ can change over time if the monopolist's forecast about $F$ evolves from a simple linear function to a non-linear function.  

The statistical properties of ${\mathcal A}({\mathcal D}_t)$ can change over time as well.   If the monopolist estimates the parameters of $F$, the second moment of the estimates may decrease as more data become available.   Our formulation of algorithm ${\mathcal A}$ is sufficiently general to cover the construction of new specifications over time as well as the estimation of the parameters.

We impose several conditions on ${\mathcal A}$.   We sometimes write ${\mathcal A}({\mathcal D}_t : F)$ for $F\in {\mathcal F}$ to emphasize that the data is generated by the true $F$.

We impose a series of conditions for a feasible algorithm.

\medskip\noindent {\bf A1.}  Translating function. \\
Since $\theta\in\Theta$ can be arbitrary, ${\mathcal A}$ must accompany a ``translating function''
\[
  \varphi({\mathcal A}({\mathcal D}_t)) =(\varphi_p({\mathcal A}({\mathcal D}_t)),\varphi_q({\mathcal A}({\mathcal D}_t))) =(p_t,q_t)
\]
which translates $\theta$ into the (human-readable) forecast about the optimal price $p_t$ and the expected demand $q_t$ at the optimal price (thus, the maximized expected profit).   For example, if ${\mathcal A}({\mathcal D}_t)$ is the empirical distribution of $F$, then the translation function calculates the estimated optimal price and the expected quantity (thus, estimated expected profit).

Since the goal of the monopolist is to find the optimal price, it is natural that ${\mathcal A}({\mathcal D}_t)$ includes information about the forecast about the optimal price conditioned on ${\mathcal D}_t$.    On the other hand, it is a substantive restriction to require the algorithm should be able to provide information to forecast the expected quantity (and, therefore, the expected profit).   We require that the algorithm should be able to explain how to calculate the optimal price and what would be the consequence of following the forecast from the algorithm.   

\medskip\noindent{\bf A2.}  Lipschitz.\\
We assume that the translating function $\varphi$ is Lipschitz continuous. A slight change in the forecast should not drastically change the forecast of the optimal price and the expected quantity.

\medskip\noindent{\bf A3.} Recursive algorithm.\\
${\mathcal A}$ is a recursive algorithm, if $\exists\Psi$ such that
\[
{\mathcal A}({\mathcal D}_{t+1}) =\Psi ( {\mathcal A}({\mathcal D}_{t}), D_t)
\]
where ${\mathcal D}_{t+1}$ is obtained by concatenating ${\mathcal D}_t$ to $D_t$.
Instead of the entire history, the algorithm needs to remember only the most recent forecast from the algorithm and the market data.   The cost of storing past information is minimized.

\medskip\noindent{\bf A4.}  Uniform learnability. \\
For $F\in {\mathcal F}$, define
$(b^*(F),q^*(F))$ as
\[
b^*(F)=\arg\max_p p(1-F(p))
\]
and $q^*(F)=1-F(b^*(F))$.   Thanks to the increasing hazard rate property, the optimization problem admits a unique solution $b^*(F)$.

\begin{definition} \label{df: general PAC}   ${\mathcal A}$ uniformly learns ${\mathcal F}$ if $\forall\mu>0$, $\forall\lambda\in (0,1)$, $\exists \Tbar(\mu,\lambda)$ such that 
\begin{equation}
  \Prob\left( \left|  \varphi({\mathcal A}(D_{\Tbar(\mu,\lambda)})) -(b^*(F),q^*(F)) 
      \right| \ge  \mu\right) \le \lambda
  \label{eq: general PAC}
\end{equation}
where $\Tbar(\mu,\lambda)\sim \bigO\left( -\frac{\log\lambda}{\mu^p}\right)$ for some $p>0$.
\end{definition}

Our notion of uniform learnability differs from the learnability in the conventional learning models (e.g., \citeasnoun{EvansandHonkapohja00}) in weak convergence, as we require the algorithm should produce an accurate forecast by $T(\mu,\lambda)$ uniformly over ${\mathcal F}^\eta$.   Since the problem is underspecified, the agent learns optimal decisions over ${\mathcal F}^\eta$ instead of an optimal decision for a particular $F\in {\mathcal F}^\eta$.   It is natural to impose uniformity on the learnability.

Given the accuracy bound $\mu$ and the confidence bound $1-\lambda$, the uniform learnability requires us to construct an algorithm ${\mathcal A}$ and stopping time $\Tbar(\mu,\lambda)$ such that the algorithm produces a forecast by period $\Tbar(\mu,\lambda)$ satisfying accuracy and confidence bounds.

Uniform learnability is inspired by the PAC guaranteeing property in the computer science literature (\citeasnoun{Shalev-ShwartzandBen-David14}).

\begin{definition} \label{df: strong PAC}  ${\mathcal A}$ PAC guarantees ${\mathcal F}$ if 
$\forall\mu>0$, $\forall\lambda\in (0,1)$, $\exists \Tbar(\mu,\lambda)$ such that 
\begin{equation}
  \Prob\left( \left|  \varphi({\mathcal A}(D_t)) -(b^*(F),q^*(F)) 
      \right| \ge  \mu\right) \le \lambda \qquad\forall t\ge \Tbar(\mu,\lambda)
  \label{eq: strong PAC}
\end{equation}
where $\Tbar(\mu,\lambda)\sim \bigO\left( -\frac{\log\lambda}{\mu^p}\right)$ for some $p>0$.
\end{definition}

Clearly, if ${\mathcal A}$ PAC guarantees ${\mathcal F}$, ${\mathcal A}$ uniformly learns ${\mathcal F}$.   After completing the analysis, we show how we can change the algorithm slightly to satisfy PAC guaranteeing property in Section \ref{PAC Guarantee}, explaining why we opt for uniform learnability.

Let ${\mathcal D}_t=(D_1,\ldots,D_{t-1})$ be the history at the beginning of period $t$.   Let ${\mathcal D}$ be the set of all histories.  Algorithm ${\mathcal A}$ dictates the monopolist to charge $p_t$ conditioned on ${\mathcal D}_t$.
Let $(p_t,q_t)$ be the pair of realized delivery price and realized quantity in period $t$.  The long-run average payoff from algorithm ${\mathcal A}$ is
\[
{\mathcal U}({\mathcal A})=\lim_{T\rightarrow\infty}\frac{1}{T}\Expect\sum_{t=1}^T p_t q_t
\]
if the limit exists, and 0, otherwise.\footnote{We require that the data complexity increases at the logarithmic speed as the confidence bound converges to 0.  Consequently, the probability of making a forecast outside of $\epsilon$ bound vanishes at the exponential speed as the number of data increases.    If the discount rate is sufficiently small (i.e., the discount factor is sufficiently close to 1), uniformly learning ${\mathcal A}$ generates a long run discounted average payoff close to the true long run discounted average payoff uniformly over ${\mathcal F}^\eta$. }

\subsection{Existence}

\citeasnoun{ColeRoughgarden2014} proved the existence of an algorithm that PAC guarantees ${\mathcal F}^\eta$. While the original algorithm of \citeasnoun{ColeRoughgarden2014} is defined in non-recursive form, the algorithm can easily adapt to a recursive form, which we present as the benchmark for our analysis.\footnote{The result of \citeasnoun{ColeRoughgarden2014} is stronger than what we state here because \citeasnoun{ColeRoughgarden2014} constructed an algorithm that can PAC guarantee ${\mathcal F}^0$ which includes ${\mathcal F}^\eta$. Since $\cup_{\eta>0}{\mathcal F}^\eta$ is a dense subset of ${\mathcal F}^0$, we implicitly assume that $\eta>0$ is large so that the difference between ${\mathcal F}^0$ and ${\mathcal F}^\eta$ is ``small.''}

In each period, the monopolistic seller asks buyer $i\in\{1,\ldots,N\}$ reports his valuation.   Since the truthful revelation is a dominant strategy of a buyer in period $t$, we assume that each buyer in period $t$ reports his valuation $v_{i,t}$ truthfully, because the truthful reporting is a dominant strategy.   Let $D_t=(v_{1,t},\ldots,v_{N,t})$ be the data collected by the monopolist in period $t$ and
\[
  {\hat F}_{i,t}=\frac{1}{t}\# \{ t'\le t \mid v_{i,t'}\le v\}
\]
be the empirical distribution of buyer $i$ based on $t$ observations, where $\#$ is the number of elements in the set.    Given ${\mathcal D}_t=(D_1,\ldots,D_{t-1})$  and $D_t$, the monopolistic seller updates ${\hat F}_{i,t}$ recursively according to
\[
  {\hat F}_{i,t}(v) ={\hat F}_{i,t-1}(v) +\frac{1}{t}\left[
{\mathbb I}( v_{i,t}\le v) - {\hat F}_{i,t-1}(v)    \right]  \qquad\forall v.
\]
The monopolist estimates the (normalized) aggregate demand
\[
{\hat F}_t(v)=\frac{1}{N}\sum_{i=1}^N {\hat F}_{i,t}(v) \qquad\forall v.
\]
Based on ${\hat F}_t$, the monopolists offer $p_t$
that solves
\[
p_t=\arg\max_{p\in[\plbar,\pubar]} p (1-{\hat F}_t(p))
\]
and forecasts the expected amount of sales as
\[
1-{\hat F}_t(p_t).
\]
Let ${\mathcal A}_{CR}$ be the algorithm we just described.
Let us state the main result of  \citeasnoun{ColeRoughgarden2014} adapted to our framework.

\begin{theorem} (\citeasnoun{ColeRoughgarden2014})
  ${\mathcal A}_{CR}$ PAC guarantees, and therefore uniformly learns, ${\mathcal F}^\eta$.
\end{theorem}

Since $v_t$ is i.i.d., the empirical distribution ${\hat F}_{i,t}(v)$ converges to the true distribution $F_{i}(v)$ pointwise almost everywhere, and satisfies the large deviation property so that the estimation error vanishes at an exponential rate (\citeasnoun{DemboandZeitouni98}).    The estimation method is recursive and efficient.   By the nature of non-parametric estimation, the monopolist's model of the demand is correctly specified.

\subsection{Computational Cost}

${\mathcal A}_{CR}$ assumes away two sources of computational costs.   First, ${\mathcal A}_{CR}$ presumes that the extraction of private information is free.   In the revelation game, each player reports his valuation to the mechanism.    However, the revelation game is a mathematical tool to investigate an actual trading protocol. When we buy a good from a monopolist, we are rarely asked to report our reservation value.   Instead, we pay the posted price or place a bid instead of reporting our reservation value.    To infer the underlying valuation, the monopolistic seller has to invert a buyer's strategy, which is not a trivial exercise for a boundedly rational seller who does not have a complete specification of the demand curve.   To extract private information, the monopolist has to implement a protocol to communicate with the buyers in addition to the existing rule, which can be expensive. It would be more economical if the algorithm relied only on the aggregate outcome from the trading instead of the private and micro-level data such as the reservation value of an individual buyer. An excellent example of aggregate data would be the total quantity traded and the price at which the good is delivered each period.    Price and quantity would be the natural outcomes of any trading protocol.

Second, the estimated distribution ${\hat F}_{i,t}$ is generally a non-linear function, which requires many parameters to identify.    A recursive algorithm uses the estimator from the previous round and needs to store the present forecast for the next period. It would be costly to remember many parameters in each period.
A boundedly rational seller would prefer an alternative algorithm that can PAC guarantees ${\mathcal F}$ which requires the smallest number of parameters to remember.    

We modify the preference of the monopolistic seller, incorporating the complexity cost of an algorithm. Let ${\sf A}$ be the collection of recursive algorithms satisfying ${\bf A1}-{\bf A4}$. The monopolist has a lexicographic preference over the long run average payoff and the input data complexity of an algorithm as in \citeasnoun{Rubinstein86}.    The monopolist chooses an algorithm in ${\sf A}$ which requires minimal input data for the algorithm.

Since ${\mathcal A}$ is a recursive algorithm, the input in period $t$ is $({\mathcal A}({\mathcal D}_{t-1}), D_{t})$.  Define
\[
  \dim ({\mathcal A}({\mathcal D}_{t-1}), D_{t})
=\dim ({\mathcal A}({\mathcal D}_{t-1}))+\dim (D_t)
\]
as the number of variables the algorithm needs in period $t$.   A recursive algorithm has to remember the previous period's forecast ${\mathcal A}({\mathcal D}_{t-1})$, while taking in the new data $D_t$ in period $t$.

For any ${\mathcal A}\in {\sf A}$, we need to specify time $T$ when the algorithm stops to produce the forecast, subject to the accuracy and the confidence bound.    We measure the complexity of ${\mathcal A}$ by the maximum number of variables that algorithm has to remember until the algorithm produces the final forecast at $T$.

\begin{definition}
  Complexity of ${\mathcal A}$ which stops at time $T$ is
  \[
\comp_T({\mathcal A})=\max_{1\le t\le T}  \dim ({\mathcal A}({\mathcal D}_{t-1}), D_{t}).
  \]
\end{definition}

We have to ``reserve'' sufficient memory space to complete an algorithm up front.
We use the size of ``reserved memory'' to measure the complexity of an algorithm.
We interpret ${\mathcal A}({\mathcal D}_{t-1})$ broadly as any variables that the algorithm has to remember in order to produce period $t$ forecast, including the previous forecast and the variables needed to update the forecast.

\begin{example}
In ${\mathcal A}_{CR}$, the empirical distribution of the valuations of buyers is updated in each period.   The algorithm has to keep track of the empirical distribution at the end of period $t-1$, while receiving $K$ number of valuations.   An empirical distribution at the end of period $t-1$ typically requires to remember
the equal number of observations, say $K(t-1)$.  Thus, $\dim ({\mathcal A}_{CR}({\mathcal D}_{t-1})=K(t-1)$ and $\dim D_t=K$.   If the algorithm stops at $T$, $\comp_T({\mathcal A}_{CR})=K(T-1)+K=KT$.
\end{example}

$\dim(D_t)$ measures the number of data that is necessary to produce ${\mathcal A}(D_t)$.
Even if the algorithm uses the entire memory space just one period, we need to secure the memory space throughout the operation of the algorithm.   Thus, the complexity cost is determined by the maximum number of parameters and data during the operation.    

\begin{example}
  Suppose that ${\mathcal A}$ calculates the empirical distribution in a batch mode.   Instead of updating the empirical distribution in each period in response to $K$ new observations, we can calculate the empirical distribution at the end of period $T$.  In each period $t\le T-1$, the algorithm takes $K$ observations of valuations of the buyers.  For $t\le T-1$, $\dim(D_t)=K$ and $\dim ({\mathcal A}(D_t)=0$ since the algorithm produces no forecast.   In period $T$, the algorithm calculates the empirical distribution using $KT$ observations of valuations.    But, $\dim(D_T)=KT$, because the algorithm requires to read all existing data to produce ${\mathcal A}(D_T)$.   Thus, $\comp_T({\mathcal A})\ge KT$.
\end{example}

Recall that ${\mathcal U}({\mathcal A})$ is the long run average expected payoff generated by algorithm ${\mathcal A}$.  The monopolist searches the optimal price through an algorithm
\begin{equation}
  \max_{{\mathcal A}\in {\sf A}} {\mathcal U}({\mathcal A})
  \label{eq: long run}
\end{equation}
subject to the constraint that if for any ${\mathcal A}'\in {\sf A}$,
$\max_{T\ge 1}\comp_T({\mathcal A}') < \max_{T\ge 1}\comp_T ({\mathcal A})$, then
\[
{\mathcal U}({\mathcal A}) > {\mathcal U}({\mathcal A}').
\]
Given that a recursive algorithm can achieve an optimal solution in the long run, no other algorithm with simpler input data complexity can achieve the same long-run payoff as ${\mathcal A}$.

\citeasnoun{ColeRoughgarden2014} proved that the non-parametric estimation algorithm of the distribution uniformly learns ${\mathcal F}^\eta$.   Important questions are whether a simpler algorithm, in particular a parametric estimation algorithm, can uniformly learn ${\mathcal F}^\eta$, and if one exists, we aim to identify the simplest algorithm.\footnote{We admit the possibility that more than one simplest algorithm may exist.}

We first show that any algorithm in ${\sf A}$ that can PAC guarantee ${\mathcal F}^\eta$ must take at least two observations from the market, such as a pair of the delivery price and the aggregate realized demand.

\begin{proposition}  \label{pr: lower bound}
If ${\mathcal A}\in {\sf A}$ solves \eqref{eq: long run}, then $\dim D_t \ge 2$ and $\dim{\mathcal A}({\mathcal D}_{t-1})\ge 2$, implying $\comp_T({\mathcal A})\ge 4$ $\forall T\ge 1$.
\end{proposition}

\begin{proof}
  See Appendix \ref{Proof of Proposition pr: lower bound}.
\end{proof}

\section{Construction}
\label{Construction}

We construct a recursive algorithm with $\dim({\mathcal A}({\mathcal D}_{t-1})) =\dim(D_t)=2$ $\forall t\ge 1$.

\subsection{Linear Recursive Algorithm}

$D_t=(q_t,p_t)$ is the pair of quantity $q_t$ and the price $p_t$ charged in period $t$.   Recall that $q_t$ is a random variable with $\Expect q_t=1-F(p_t)$.  We can write
\[
q_t =1-F(p_t) +\epsilon_{2,t}.
\]
$\Expect\epsilon_{2,t}=0$ and $\Expect\epsilon^2_{2,t}<\infty$ uniformly, but the actual size of the second moment can depend on $p_t$ and $F\in {\mathcal F}$.

Let
\[
{\mathcal D}_t =( D_1,\ldots,D_{t-1})
\]
be the history at the beginning of period $t$.   The monopolist assumes that the aggregate demand is a linear function:
\[
q =\beta_0+\beta_1 p
\]
and estimates $(\beta_0,\beta_1)$ according to the least square estimation over ${\mathcal D}_t$.   Let ${\mathcal H}$ be the set of all linear demand functions, parameterized by $(\beta_0,\beta_1)$.

Since $F\in {\mathcal F}^\eta$, the optimal solution $b^*(F)$ must generate a positive profit.   Thanks to the uniform bound $\eta$, there exists a compact set $K$ in the interior of ${\mathbb R}^2_+$ such that $b^*(F)\in K$ $\forall F\in {\mathcal F}^\eta$.
Thus, $(\beta_0,\beta_1)$ must be such that the optimal price and the expected quantity under the linear demand curve parameterized by $(\beta_0,\beta_1)$ must be contained in $K$.

Recall that if $(\beta_0,\beta_1)$ can support $b^^*(F)$ for some $F\in{\mathcal F}^\eta$,
\[
1-F(b^*(F)) =\frac{\beta_0}{2} \ \ \text{and} \ \ b^*(F)= -\frac{\beta_0}{2\beta_1}
\]
or equivalently,
\[
\beta_0 =2(1-F(b^*(F))   \ \ \text{and} \ \ \beta_1=-\frac{1-F(b^*(F))}{b^*(F)}.
\]
Since $(1-F(b^*(F)),b^*(F))\in K$, there exists a compact set ${\sf B}\subset (0,\infty)\times (-\infty,0)$ such that $(\beta_0,\beta_1)\in {\sf B}$ if the linear demand can support a true optimal price.    If $(\beta_0,\beta_1)\not\in {\sf B}$, then the monopolist can conclude that the estimated demand is wrong, based on what the monopolist knows.

Let ${{\mathcal H}}_{{\sf B}}\subset {\mathcal H}$ be the collection of the linear demand, which induces the market outcome in $K$.    The seller knows that the market demand curve is in ${{\mathcal H}}_{{\sf B}}$. Since ${\mathcal F}$ typically contains non-linear demand curves, ${{\mathcal H}}_{{\sf B}}$ is a misspecified model.    Nevertheless, the number of parameters the seller has to keep track of is minimal.

Since the seller does not know $(\beta_0,\beta_1)$, the seller recursively estimates the parameters using the least square estimation method while choosing the pricing rule based on the estimated linear demand curve.
Let $(\beta_{0,t-1},\beta_{1,t-1})$ be the least square estimator at the end of period $t-1$.   Given the estimated demand curve
\[
q=\beta_{0,t-1}+\beta_{1,t-1}p,
\]
the monopolist calculates the optimal price
\[
-\frac{\beta_{0,t-1}}{2\beta_{1,t-1}}
\]
but incurs an implementation error $\epsilon_{1,t}$  so that the actual price in period $t$ is
\[
p_t = -\frac{\beta_{0,t-1}}{2\beta_{1,t-1}}+\epsilon_{1,t}
\]
where $\epsilon_{1,t}$ is i.i.d. with $\Expect\epsilon_{1,t}=0$ and $\Expect\epsilon^2_{1,t}=\sigma^2_1$.
We interpret $\epsilon_{1,t}$ as the implementation error or the small experimentation by the monopolist seller.    We choose $\epsilon_{1,t}$ from a small interval $[-\epsilon,\epsilon]$ according to a fixed distribution, say the uniform distribution.   We control the size of $\epsilon>0$ to achieve the desired level of accuracy of the algorithm.

The monopolist forecasts that the sales quantity will be
\[
  \beta_{0,t-1}+\beta_{1,t-1}\left[ -\frac{\beta_{0,t-1}}{2\beta_{1,t-1}}+\epsilon_{1,t}\right]
  =\frac{\beta_{0,t-1}}{2} +\beta_{1,t-1}\epsilon_{1,t}
\]
but the actual quantity is
\[
1-F\left( -\frac{\beta_{0,t-1}}{2\beta_{1,t-1}}+\epsilon_{1,t}  \right)+\epsilon_{2,t}.
\]
The forecasting error is
\[
\phi(\beta_{t-1},\epsilon_t)
=1-F\left( -\frac{\beta_{0,t-1}}{2\beta_{1,t-1}}\right)+\epsilon_{1,t}  +\epsilon_{2,t} -
\frac{\beta_{0,t-1}}{2}-\beta_{1,t-1}\epsilon_{1,t}
\]
where $\beta_{t-1}=(\beta_{0,t-1},\beta_{1,t-1})$ and $\epsilon_t=(\epsilon_{1,t},\epsilon_{2,t})$.

Since the quantity must be in the closed interval $[0,1]$, we first take care of the cases of ``corner solution'' before moving to ``interior solution.''
If actual quantity
\[
  q_t=1-F\left( -\frac{\beta_{0,t-1}}{2\beta_{1,t-1}}\right)+\epsilon_{1,t}
\]
is at the boundary, we directly update $(\beta_{0,t-1},\beta_{1,t-1})$.
If $q_t=0$, then the algorithm concludes that the forecast price was too high and adjusts accordingly:
\[
\beta_{0,t}=\beta_{1,t}-a \  \ \text{and} \ \ \beta_{1,t}=\beta_{1,t-1}
\]
so that
\[
b_{t+1}=-\frac{\beta_{0,t}}{2\beta_{1,t}}=-\frac{\beta_{0,t-1}}{2\beta_{1,t-1}}+\frac{a}{2\beta_{1,t-1}}=b_t+\frac{a}{2\beta_{1,t-1}} < b_t,
\]
where small $a>0$ is a parameter of the algorithm that will be determined to satisfy the precision and the confidence requirements.   Similarly, if $q_t=1$, 
\[
\beta_{0,t}=\beta_{1,t}+a \  \ \text{and} \ \ \beta_{1,t}=\beta_{1,t-1}
\]
so that
\[
b_{t+1}=b_t -\frac{a}{2\beta_{1,t-1}} > b_t.
\]


If $0< q_t <1$, then 
\begin{equation}
  \left[
\begin{matrix}
  \beta_{0,t} \\
  \beta_{1,t} 
\end{matrix}
  \right] =
    \left[
\begin{matrix}
  \beta_{0,t-1} \\
  \beta_{1,t-1} 
\end{matrix}
\right]  + a_t R^{-1}_{t-1}
  \left[
\begin{matrix}
  1 \\
  -\frac{\beta_{0,t-1}}{2\beta_{1,t-1}} +\epsilon_{1,t}
\end{matrix}
  \right] \phi(\beta_{t-1},\epsilon_t)
  \label{eq: LSE}
\end{equation}
where
\[
R_{t-1} =
  \left[
\begin{matrix}
  1  &    -\frac{\beta_{0,t-1}}{2\beta_{1,t-1}} \\
  -\frac{\beta_{0,t-1}}{2\beta_{1,t-1}}  & \left(-\frac{\beta_{0,t-1}}{2\beta_{1,t-1}}\right)^2+\sigma_1^2
\end{matrix}
  \right].
\]
Since the monopolist designs the size of the experiments, the variance $\sigma^2_1$ of $\epsilon_{1,t}$ is a known parameter.

We need to impose a bound to $(\beta_{0,t},\beta_{1,t})$ to keep the estimator within a compact set.   Let ${\mathcal B}$ be a compact convex set that contains ${\sf B}$ in the interior of ${\mathcal B}$ so that the Hausdorff distance between ${\mathcal B}$ and ${\sf B}$ is positive.    If $(\beta_{0,t},\beta_{1,t})\not\in {\mathcal B}$, then the seller can conclude that the estimator is out of the line and needs to adjust the estimator by pushing it back to ${\sf B}$.\footnote{This mapping is known as the projection facility in the literature of the stochastic approximation (\citeasnoun{KushnerandYin97}).}

We modify the baseline updating scheme to construct the formal updating scheme for $(\beta_{0,t}\beta_{1,t})$.
\[
  \left[
\begin{matrix}
  \beta_{0,t} \\
  \beta_{1,t} 
\end{matrix}
  \right] =
    \left[
\begin{matrix}
  \beta_{0,t-1} \\
  \beta_{1,t-1} 
\end{matrix}
\right]  + a_t R^{-1}_{t-1}
  \left[
\begin{matrix}
  1 \\
  -\frac{\beta_{0,t-1}}{2\beta_{1,t-1}} +\epsilon_{1,t}
\end{matrix}
  \right] \phi(\beta_{t-1},\epsilon_t)
\]
if the right hand side is in ${\mathcal B}$.   Otherwise, 
$(\beta_{0,t},\beta_{1,t}) =(\beta_0,\beta_1)\in {\sf B}$ for some fixed $(\beta_0,\beta_1)$ in the interior of ${\sf B}$.   
To simplify notation, we write
\begin{equation}
\varphi_{t-1}\equiv\varphi(\beta_{t-1},\epsilon_t)=R^{-1}_{t-1}
  \left[
\begin{matrix}
  1 \\
  -\frac{\beta_{0,t-1}}{2\beta_{1,t-1}} +\epsilon_{1,t}
\end{matrix}
  \right] \phi(\beta_{t-1},\epsilon_t).
\label{eq: varphi def}
\end{equation}
Treating the estimated demand curve
\[
q =\beta_{0,t-1}+\beta_{1,t-1}p
\]
as the actual demand curve, the seller sets the price
\[
p_t=-\frac{\beta_{0,t-1}}{2\beta_{1,t-1}}+\epsilon_{1,t}
\]
where $\epsilon_{1,t}$ is an i.i.d. white noise
uniformly distributed over $[-\epsilon,\epsilon]$ for a fixed $\epsilon>0$ so that
$\Expect \epsilon_{1,t}=$ and $\Expect\epsilon^2_{1,t}=\sigma^2_1$.   Given $p_t$, the quantity in period $t$
\[
q_t=1-F(p_t)+\epsilon_{2,t}
\]
is realized.   Using $(q_t,p_t)$, the seller updates $(\beta_{0,t-1},\beta_{1,t-1})$ to $(\beta_{0,t},\beta_{1,t})$.
The translating function 
\[
  \phi(\beta_{0,t-1},\beta_{1,t-1})=
  \left(
    -\frac{\beta_{0,t-1}}{2\beta_{1,t-1}},
    1-F\left(    -\frac{\beta_{0,t-1}}{2\beta_{1,t-1}}+\epsilon_{1,t}\right)
  \right)
\]
maps the estimated linear demand curve into the mean forecast price and quantity.

We need to specify one remaining parameter of the algorithm: gain function $a_t$.  Instead of $a_t=1/t$ as in the recursive least square estimation, we set $a_t=a>0$\footnote{The resulting algorithm is often called the constant gain least square estimation algorithm. } which we already used to control the algorithm when $q_t=0$ or $q_t=1$. 

 We have to choose $a$ to achieve the desired level of accuracy and confidence.   Let ${\mathcal A}_a$ be the recursive algorithm with constant gain $a_t=a>0$ $\forall t\ge 1$.  $\forall a\ge 0$, algorithm ${\mathcal A}_a$ produces $\beta_t$ following history ${\mathcal D}_t$, where
\[
{\mathcal D}_t=( (q_1,p_1),\ldots,(q_{t-1},p_{t-1}))
\]
is the sequence of aggregate market outcomes up to period $t-1$.  The constructed algorithm is recursive:  ${\mathcal A}_a({\mathcal D}_t)=\beta_t$ is the output of the algorithm based on $D_t$ and ${\mathcal A}_a({\mathcal D}_{t-1})$.   The input complexity 
\[
\dim ({\mathcal A}_a({\mathcal D}_{t-1}))=2 \ \text{and} \ \dim(D_{t-1})=2 \qquad\forall t\ge 1,
\]
so that $\comp_T({\mathcal A}_a)=4$ $\forall T\ge 1$.

To emphasize that the optimal price $b^*$ is a function of the underlying (aggregate) distribution $F$, we sometimes write $b^*(F)$ instead of $b^*$.   Let $\beta^*(F)$ be the pair of estimators that induce the optimal price for $F$:
\[
b^*(F)=-\frac{\beta^*_0(F)}{2\beta^*_1(F)} \ \ \text{and} \ \ q^*(F)=1-F(b^*(F)).
\]

\begin{theorem}  \label{thm:mainresult}
$\forall\mu>0$,  $\exists \tau(\mu)>0$, $\rho>0$ and ${\overline a}>0$ such that $\forall a\in (0,{\overline a})$, if $T=\lceil \frac{\tau(\mu)}{a}\rceil$ 
$\forall a\in (0,{\overline a})$,
\[
\Prob\left( \left| \varphi\left({\mathcal A}_a({\mathcal D}_{T(a,\mu)})\right)-(b^*(F),q^*(F))\right|>4\mu \right) \le e^{-\rho T(a,\mu)}
\qquad\forall F\in {\mathcal F}^\eta
\]
and $\tau(\mu)\sim -\log \mu$ for small $\mu>0$ and
${\overline a}=\bigO\left( -\frac{\mu}{\log\mu} \right)$.
\end{theorem}

\begin{proof}
See Appendix \ref{Proof of Theorem thm:mainresult}.
\end{proof}

For fixed $\lambda>0$ less than 1, we can choose ${\overline a}>0$ sufficiently small so that
\[
e^{-\frac{\rho \tau(\mu)}{{\overline a}}} = \lambda
\]
and therefore,
\[
\frac{\tau(\mu)}{{\overline a}} = -\frac{\log\lambda}{\rho}.
\]
Since
\[
T({\overline a},\mu)=\left\lceil \frac{\tau(\mu)}{{\overline a}} \right\rceil =\left\lceil -\frac{\log\lambda}{\rho}\right\rceil,
\]
$T({\overline a},\mu)$ increases at the logarithmic speed with respect to $1/\lambda$.
An important observation is that the estimator's accuracy depends only on $\tau(\mu)$, and the approximation error vanishes at the linear rate of $a\rightarrow 0$.   Thus, the minimum number of the time steps to satisfy the accuracy requirement increases at the rate of $-\frac{\log\mu}{a}$.   Theorem \ref{thm:mainresult} implies that ${\mathcal A}_a$ uniformly learns ${\mathcal F}^\eta$.   

\section{Heuristics}
\label{Heuristics}

Instead of a formal proof, let us provide a heuristic explanation about how to prove Theorem \ref{thm:mainresult}. An important implication of the increasing hazard rate property is that the optimal price $b^*(F)$ associated with $F$ is unique and completely determined by the first order condition
\begin{equation}
\frac{1-F(b^*(F))}{f(b^*(F))}-b^*(F) =0.
\label{eq: first order condition}
\end{equation}
Suppose that the monopolist knows $F\in {\mathcal F}^\eta$ for some $\eta>0$, even though the monopolist does not know precisely what $F$ is and, therefore, does not know $b^*(F)$.
Recursively estimating $F$ non-parametrically requires the algorithm to remember the empirical distribution in each round.   If $F$ is non-linear, the number of parameters to remember can be large and increase as the monopolist observes more data.
To save the cost of memory, the monopolist assumes the market demand curve is
\[
q =\beta_0+\beta_1 p
\]
and estimates $(\beta_0,\beta_1)$ recursively according to the least square estimation as defined by \eqref{eq: LSE}. Since $F\in{\mathcal F}^\eta$ is typically non-linear, the linear model of the monopolist is misspecified.    On the other hand, linear demand is the simplest model to formulate a non-trivial downward-sloping market demand function.

Let $(\beta_{0,t-1},\beta_{1,t-1})$ be the estimator at the end of period $t-1$.
The monopolist updates the estimator to  $(\beta_{0,t},\beta_{1,t})$ after she observes price $p_t$ and quantity $q_t$ in period $t$.   The optimal price associated with linear demand $q=\beta_{0,t-1}+\beta_{1,t-1}p$ is
\[
b_t =-\frac{\beta_{0,t-1}}{2\beta_{1,t-1}}.
\]
Given $(\beta_{0,t-1},\beta_{1,t-1})$, the monopolist charges price according to 
\begin{equation}
p_t=b_t+\epsilon_{1,t}.
\label{eq: random price}
\end{equation}
The monopolist experiments a little by adding a small noise $\epsilon_{1,t}$ which is drawn from a binomial distribution over $\{-\epsilon,\epsilon\}$ for small $\epsilon>0$.\footnote{Any distribution with a mean 0 and a small variance will do.   The binary distribution is the worst case to obtain the confidence bound. }
Given $b_t$, the quantity forecast is
\[
\beta_{0,t-1}+\beta_{1,t-1}p_t
\]
but the actual demand is
\[
q_t=1-F(p_t)+\epsilon_{2,t}.
\]
The forecasting error is 
\[
1-F(p_t) -\left[ \beta_{0,t-1}+\beta_{1,t-1}p_t\right]+\epsilon_{2,t}
\]
is
\[
\phi(\beta_{t-1},\epsilon_t) =1-F(b_t+\epsilon_{1,t}) -\left[ \beta_{0,t-1}+\beta_{1,t-1}b_t +\beta_{1,t-1}\epsilon_{1,t}\right]
\]
as in \eqref{eq: LSE}.  
Let us write \eqref{eq: LSE} more concisely as
\[
  \beta_t=\beta_{t-1}+a \varphi(\beta_{t-1},\epsilon_t)
\]
where $\varphi$ is defined in \eqref{eq: varphi def}.
The monopolist adjusts $(\beta_{0,t},\beta_{1,t})$ according to the value of $\phi_t$ and $(\beta_{0,t-1},\beta_{1,t-1})$.

For fix $\tau>0$, define
\[
T(a,\tau)=\left\lceil \frac{\tau}{a}\right\rceil.
\]
Following the control theory literature (\citeasnoun{KushnerandYin97}), we interpret $\tau$ as the (clock) time, and $a>0$ as the time interval between two adjacent observations of $(q_t,p_t)$ and $(q_{t-1},p_{t-1})$.   Thus, $T(a,\tau)$ as the number of time steps that can be ``squeezed'' into $\tau>0$ (clock) time. 
We are interested in the dynamics of $\beta_t$ over small $\tau>0$ when the monopolist can observe data very frequently:
\[
\frac{\beta_{t+T(a,\tau)}-\beta_t}{\tau}
\]
with $\beta_t=\beta_0$,
whose property can be inferred by taking limits:\footnote{Under the conditions we have imposed, a limit exists (cf. \citeasnoun{KushnerandYin97}).}
\[
\lim_{\tau\rightarrow 0}\lim_{a\rightarrow 0}\frac{\beta_{t+T(a,\tau)}-\beta_t}{\tau}.
\]
Using the recursive nature of the algorithm, we can write
\[
\beta_{t+T(a,\tau)}-\beta_t =\frac{\tau}{T(a,\tau)}\sum_{k=1}^{T(a,\tau)}\varphi_{t+k}
\]
which implies
\begin{eqnarray*}
\frac{\beta_{t+T(a,\tau)}-\beta_t}{\tau}
& = &
\frac{1}{T(a,\tau)}\sum_{k=1}^{T(a,\tau)}\varphi_{t+k} \\
& = & \frac{1}{T(a,\tau)}\sum_{k=1}^{T(a,\tau)}\Expect_{t+k-1}\varphi_{t+k}
+\frac{1}{T(a,\tau)}\sum_{k=1}^{T(a,\tau)}\varphi_{t+k}-\Expect_{t+k-1}\varphi_{t+k}     
\end{eqnarray*}
where $\varphi_{t+k}$ is defined in \eqref{eq: varphi def}.   Define
\begin{equation}
\xi_{t+k}=\phi_{t+k}-\Expect_{t+k-1}\varphi_{t+k}
\label{eq: xitk}
\end{equation}
which is a martingale difference.

Let us examine the first term when $a\rightarrow 0$ for a small fixed $\tau>0$.
$\beta_{t+T(a,\tau)}-\beta_t$ is a discounted average of $\{\phi_{t+k}\}_{k=1}^{T(a,\tau)}$, which converges to its simple average as $a\rightarrow 0$.   We can approximate
\[
\lim_{a\rightarrow 0} \frac{\beta_{t+T(a,\tau)}-\beta_t}{\tau} \simeq \Expect \left[ \varphi_{t} \mid \beta_t=\beta\right]
\]
or more concisely as
\[
{\dot b}=\Expect \left[ \varphi_{t} \mid \beta_t=\beta \right].
\]
To derive the formula of the right hand side, let us fix $(\beta_{0,t},\beta_{1,t})=(\beta_0,\beta_1)\equiv \beta$ and calculate the expected value of the estimator in the ``next period'' if the monopolist chooses the estimator to minimize the forecasting error.   Given $(\beta_0,\beta_1)$, the next period's quantity $q'$ and price $p'$ are
\[
  (q',p')=\begin{cases}
    \left(b+\epsilon, 1-F\left(b+\epsilon \right)  \right)
    & \text{with probability} \ 0.5 \\
    \left(b-\epsilon, 1-F\left(b-\epsilon \right)  \right)
    & \text{with probability} \ 0.5.
\end{cases}
\]
where
\begin{equation}
b = -\frac{\beta_0}{2\beta_1}.
\label{eq: best response}
\end{equation}
The monopolist chooses the new coefficients $(\beta'_0,\beta'_1)$ to fit the observed data best,  the new regression line must pass through
$\left(b+\epsilon, 1-F\left(b+\epsilon \right)  \right)$  and
$\left(b-\epsilon, 1-F\left(b-\epsilon \right)  \right)$.    A simple calculation shows
\begin{eqnarray*}
\beta'_0 & = & 1- F\left(b\right) +bf(b) \\ 
\beta'_1 & = & -f(b)
\end{eqnarray*}
modulo linear approximation error at the order of $\epsilon^2$.   
The stochastic approximation theory (\citeasnoun{KushnerandYin97}) shows that the asymptotic properties of the mean of $(\beta_{0,t},\beta_{1,t})$ is dictated by the dynamic properties of the associated ordinary differential equation (ODE)
\begin{eqnarray}
{\dot \beta}_0 & = & \beta'_0- \beta_0 = 1-F(b)+b f(b) -\beta_0 \label{eq: error ODE} \\
{\dot \beta}_1 & = & \beta'_1 -\beta_1 = -f(b)-\beta_1. \nonumber
\end{eqnarray}
Since \eqref{eq: best response} holds in every period, we take the time derivative on both sides of the equality to have
\[
{\dot b} =-\frac{1}{2\beta_1} \left( {\dot \beta}_1 +2b {\dot \beta_0}\right).
\]
After substituting ${\dot\beta}_1$ and ${\dot\beta}_0$ by \eqref{eq: error ODE}, we have
\begin{equation}
  {\dot b}= -\frac{f(b)}{2\beta_1} \left[   \frac{1-F(b)}{f(b)}-b\right].
  \label{eq: first ODE}
\end{equation}
Since the demand curve $1-F(p)$ is strictly decreasing, $\beta_1<0$.   Thus, $-\frac{f(b)}{2\beta_1}>0$.   The term inside the bracket has a unique solution $b^*(F)$, the profit-maximizing price for (actual) distribution $F$.    By the increasing hazard rate property, the term in the bracket is strictly decreasing for $b$, which makes $b^*(F)$ a stable stationary solution of \eqref{eq: first ODE}.   The stochastic approximation implies that the least square learning algorithm weakly converges to $b^*(F)$ (\citeasnoun{KushnerandYin97}).

So far, we only prove the convergence ``pointwise'' for $F$, allowing the number of data needed to achieve the desired level of accuracy can depend on $F$.   To show the uniform convergence over ${\mathcal F}^\eta$, we need to do additional work.   Since $\frac{1-F(b^*(F))}{f(b^*(F))}-b^*(F)=0$,
\[
\frac{1-F(b)}{f(b)}-b = \frac{1-F(b)}{f(b)}-b -\left(\frac{1-F(b^*(F))}{f(b^*(F))}-b^*(F)\right).
\]
With the increasing hazard rate property,
\[
\frac{1-F(b)}{f(b)}-b -\left(\frac{1-F(b^*(F))}{f(b^*(F))}-b^*(F)\right)\le -(b-b^*(F)). 
\]
We can show that $\exists c>0$ such that\footnote{If we assume that $\inf_pf(p)>0$ uniformly, the proof is straightforward.   Without this assumption, we need some additional work (Lemma \ref{lm: c exists}).}
\[
{\dot b}={\dot {(b-b^*(F))}} \le -c (b-b^*(F))<0
\]
if $b > b^*(F)$.    Similarly, if $b< b^*(F)$, then
\[
  {\dot b} \ge -c (b-b^*(F))>0
\]
for any $F\in {\mathcal F}^\eta$.
Since $b^*(F)\le[\plbar,\pubar]$ $\forall F\in{\mathcal F}^\eta$, the initial condition of the ordinary differential equation can be selected from a compact set.   Since the distance $|b-b^*(F))|$ vanishes uniformly at the order of $e^{-c \tau}$, it takes $\bigO(-\log\mu)$ amount of time for $b$ to enter the $\mu$ neighborhood of $b^*(F)$.   From this observation, we prove that the amount of data to approximate the optimal price is uniform over ${\mathcal F}^\eta$.    We can also show that the number of data to achieve $\mu$ accuracy increases at the polynomial rate of $1/\mu$.

To calculate the confidence bound, we need to examine the distribution of
\[
\frac{1}{T(a,\tau)}\sum_{k=1}^{T(a,\tau)} \xi_{t+k}
\]
where $\xi_{t+k}$ is defined as \eqref{eq: xitk}.
The average converges to 0, but we need to find $\rho>0$ such that
\[
  \Prob\left(
\left| \frac{1}{T(a,\tau)}\sum_{k=1}^{T(a,\tau)} \xi_{t+k}   \right| >\mu
    \right)\le e^{-\rho T(a,\tau)} \qquad\forall F\in {\mathcal F}^\eta.
\]
This part of the exercise is to calculate the tail portion of the probability distribution of $\xi_{t+k}$.   For a fixed $F$, the existence of $\rho>0$ can be proved by the large deviation properties (\citeasnoun{DemboandZeitouni98}) of a recursive algorithm (\citeasnoun{DupuisandKushner89}).   Our exercise is more challenging because we are searching for $\rho>0$ {\em uniformly} over the set of feasible distributions.

The algorithm of \citeasnoun{ColeRoughgarden2014} uses the valuation of the buyers, which are drawn independently from the same distribution.  In that case, we could invoke the large deviation property of the IID sample average, such as Chernoff's bound, to prove that the tail probability vanishes at the exponential rate uniformly over ${\mathcal F}^\eta$.

In our case, the algorithm uses $(q_t,p_t)$ which is not IID (not even martingale), because of $p_t=b_t+\epsilon_{1,t}$ and $b_t$ is responding to the realized sequence of data.  The data generating process is endogenous, making the stochastic process $(q_t,p_t)$ highly non-stationary.   As a result, $\xi_{t+k}$ is not IID but a martingale difference. We need to invoke the Azuma-Hoeffding-Bennett inequality (\citeasnoun{DemboandZeitouni98}) to calculate the uniform exponential rate for all feasible distributions of buyer's valuation, which proves that our algorithm is efficient (\citeasnoun{Shalev-ShwartzandBen-David14}).

Through a sequence of linear demand curves, the seller \emph{can} still learn how to choose the optimal price at an exponential rate, \emph{even though} they may be grossly misspecified in terms of the demand curves that they consider.   Thus, within a polynomial time, the monopolistic seller behaves as if he knows the actual demand curve.    Even though a non-parametric estimation of the demand curve is feasible, as in \citeasnoun{ColeRoughgarden2014}, the monopolist chooses to use a simple yet misspecified model of the demand curve to choose his price to save the computational cost.  

\section{Numerical Experiments}
\label{Numerical Experiments}

While our algorithm ${\mathcal A}_a$ has the same asymptotic properties as the algorithm ${\mathcal A}_{CR}$ (\citeasnoun{ColeRoughgarden2014}), we need to rely on the numerical analysis to compare the performance of the algorithm with a finite number of data.    It is possible to write the algorithm of \citeasnoun{ColeRoughgarden2014} in a recursive form.
Their algorithm is best implemented in an off-line form because a recursive formulation of the algorithm of \citeasnoun{ColeRoughgarden2014} requires storing the empirical distribution of the valuation after the $t$ period, which almost always requires remembering $t$ data points.  

Instead of ${\mathcal F}^\eta$, we generate the set of feasible distributions over the buyer's valuations from a truncated Gaussian distribution at the mean, where the Gaussian distribution has a mean of 10.   By changing the standard deviation of the Gaussian distribution, we can generate different distributions of the valuations, each of which satisfies the increase hazard rate property and other regularity properties of ${\mathcal F}^\eta$.    We select 5000 different standard deviations, ranging from 11 to 16, with an increment of 0.001.   For each distribution of valuations, we calculate the actual optimal price.   As the standard deviation becomes larger, so does the actual optimal price.   We assume that there are 100 buyers whose reservation value is drawn from the same distribution $F$. The actual amount of delivery for a price is thus random, whose mean is $1-F(p)$.  

For each distribution, we run our algorithm ${\mathcal A}_a$ with $a=0.0001$ and $\epsilon=0.75$ (the size of price perturbation) for $T=300,000$ rounds.   At the end of $T$ rounds, we calculate the forecast price and compare it to the optimal price $b^*(F)$ of the actual distribution to calculate the forecasting error for each $F$.   We generate the distribution of forecasting error over the set of feasible distribution functions.\footnote{The distribution can be approximated as a Gaussian distribution which is a solution of the Ornstein-Ullenbeck equation (cf.\citeasnoun{KushnerandYin97}).}   For a small $a>0$, the distribution of the forecasting error has a mean 0 with a small variance. We found the mean is $0.0081$, and the variance is $0.0027$.      Because of the linear approximation of a non-linear demand curve, the linear approximation error contributes to the forecasting error, which vanishes as we reduce the size of price perturbation $\epsilon>0$.     We plot the distribution of forecasting error as a blue bar in a histogram in Figures \ref{fig: same size} and \ref{fig: five times} in blue bars.  The horizontal axis shows the forecasting error at $T$.
The heights of each part represent the number of distributions whose forecasting error is within the $0.01$ neighborhood of each grid.

We calculate the optimal price forecast from \citeasnoun{ColeRoughgarden2014} by drawing $K$ numbers of valuations each period for $T$ rounds to calculate the empirical distribution (thus, $KT$ samples of valuations) and the optimal price from the empirical distribution.    We calculate the forecast error, and plot the distribution  in Figure \ref{fig: same size} and \ref{fig: five times} in orange color.    In Figure \ref{fig: same size}, $K=2$ so that ${\mathcal A}_{CR}$ uses the same number of data as ${\mathcal A}_a$ in $T$ rounds.   In Figure \ref{fig: five times}, $K=10$ so that ${\mathcal A}_{CR}$ uses five times as many data as ${\mathcal A}_a$.

\begin{table}[h]
  \centering
  \begin{tabular}{lSS}
    \toprule
    \multirow{2}{*}{Per Period} &
      \multicolumn{2}{c}{ ${\mathcal A}_{CR}$ } \\
      & {Mean} & {Variance}   \\
      \midrule
    {$2$} & {$-0.0039$} & {$0.0102$}   \\
    {$4$} & {$-0.0002$} & {$0.0064$}   \\
    {$6$} & {$-0.0018$} & {$0.0048$}   \\
    {$8$} & {$0.0001$} & {$0.0041$}   \\
    {$10$} & {$0.0003$} & {$0.0035$}  \\
    \bottomrule
  \end{tabular}
  \caption{The first column reports the number of valuations the algorithm of Cole and Roughgarden [2014] takes per period.   For ${\mathcal A}_{CR}$, we draw different number of samples, which changes the mean and the variance of the forecasting errors.   The true profit maximizing price is roughly ranging from 10 to 20.}
  \label{tb: summary}
\end{table}

Table \ref{tb: summary} reports means and variances from the numerical exercises. 
Because the non-parametric estimation method in \citeasnoun{ColeRoughgarden2014} does not incur any linear approximation error, the mean of the forecasting error is closer to 0.   Interestingly, the variance of forecasting errors is larger for the same number of data as ${\mathcal A}_a$, which takes two data points (price and quantity) in each period.   If \citeasnoun{ColeRoughgarden2014} allows the algorithm to receive ten valuation reports (which is five times as much as data ${\mathcal A}_a$ uses), the variance becomes comparable to the variance of the forecasting error of ${\mathcal A}_a$.

Figure \ref{fig: same size} reports the distribution of the forecasting errors of ${\mathcal A}_a$ in blue and ${\mathcal A}_{CR}$ in orange when ${\mathcal A}_{CR}$ takes two valuation reports per period.    Note that the distribution of forecasting errors of $CR$ is spread out more than $CL$.    Figure \ref{fig: five times} reports the distribution of forecasting errors when $CR$ takes ten reported values in each period.    The two distributions of the forecasting errors become closer, as Table \ref{tb: summary} indicates.

\begin{figure}[ht]
\centering
\includegraphics[width=0.8\linewidth,height=3in]{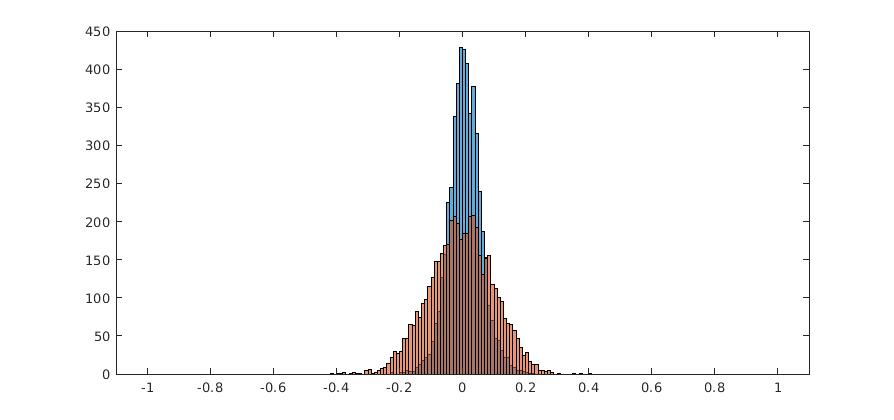}
\caption{Blue bars represent the density of forecasting errors of our algorithm ${\mathcal A}_a$ and orange ones represent the distribution of errors of the algorithm of Cole and Roughgarden [2014].  If the algorithm of Cole and Roughgarden [2014] takes 2 reports per period, the variance is 3.5 times as large as that of ${\mathcal A}_a$. }
  \label{fig: same size} 
\end{figure}

\begin{figure}[ht]
  \centering
  \includegraphics[width=0.8\linewidth,height=3in]{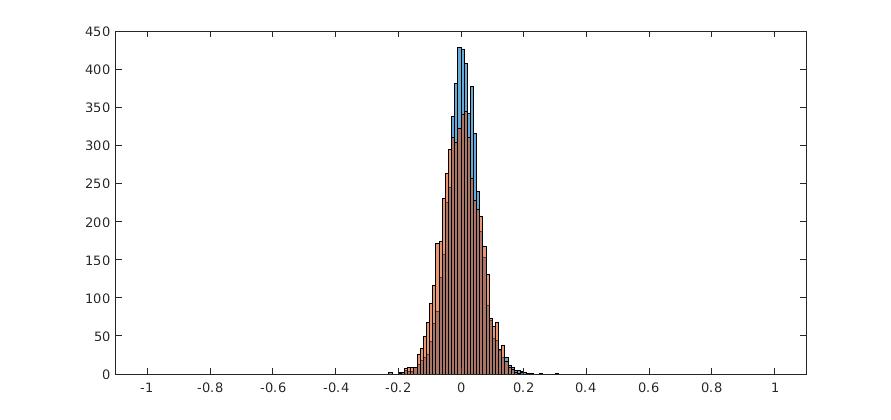}
  \caption{If the algorithm of Cole and Roughgarden [2014] takes 10 reports per period, the variance is comparable to that of ${\mathcal A}_a$.  The distributions of the forecasting errors almost overlap with each other. }
  \label{fig: five times}
\end{figure}

The performance of our algorithm ${\mathcal A}_a$ is comparable to that of \citeasnoun{ColeRoughgarden2014}.   We can reduce the linear approximation error by reducing the price perturbation, choosing a smaller $\epsilon>0$.   We can reduce the forecasting error variance by reducing the gain $a>0$.    On the other hand,  \citeasnoun{ColeRoughgarden2014} can reduce the variance of the forecasting error simply by collecting more data about the valuations.

The key advantages of our algorithm are the source of data and the recursive nature of the algorithm.   Instead of extracting the buyer's private information, our algorithm relies on the public outcome of the market in each period.    In each period, the algorithm needs to remember just two parameters of the linear demand. In contrast, in \citeasnoun{ColeRoughgarden2014}, the algorithm has to remember the empirical distribution, which essentially requires remembering the total number of reported valuations.

\subsection{PAC Guarantee}
\label{PAC Guarantee}

We can modify ${\mathcal A}_a$ to satisfy the PAC guaranteeing property \eqref{eq: strong PAC}.  Uniform learnability requires that the algorithm weakly converges to the true solution uniformly over ${\mathcal F}$, while PAC guaranteeing property requires uniform convergence in probability.

We need to modify the algorithm so that ${\mathcal A}({\mathcal D}_t)$ converges to the true optimal solution in probability.   Instead of a positive constant gain function $a>0$, we can choose $a=\frac{1}{t^\omega}$ where $\omega \in (0,1)$.  Let ${\mathcal A}_{\frac{1}{t^\omega}}$ be the modified algorithm often called a decreasing gain algorithm.

Invoking a well known result in the stochastic approximation (\citeasnoun{DupuisandKushner89}), we can show that \eqref{eq: strong PAC} holds ``pointwise'' with respect to $F\in {\mathcal F}$.    Following the same argument as the proof of Theorem \ref{thm:mainresult}, we obtain the uniform convergence to show that 
${\mathcal A}_{\frac{1}{t^\omega}}$ PAC learns ${\mathcal F}^\eta$ where
$T(\mu,\lambda)\sim \bigO\left( -\frac{\log\lambda}{\mu^{3-2\omega}}\right)$.

While \eqref{eq: strong PAC} is a stronger notion of convergence than \eqref{eq: general PAC}, we chose ${\mathcal A}_a$ over ${\mathcal A}_{\frac{1}{t^\omega}}$ for two practical reasons.   First, we have to specify the stopping time of an algorithm to produce a forecast which the monopolist can use.    When the algorithm reaches $T(\mu,\lambda)$ rounds, the two algorithms produce a strong forecast with minimal difference.    Second, the changing gain function $1/t^\omega$ could be considered as an additional parameter the algorithm should remember.   Thus, one might argue that our algorithm has three, instead of two, parameters to remember.

\section{Concluding Remarks}
\label{Concluding Remarks}

In the monopolistic market, the buyer's decision problem is simple.   The truthful revelation of his valuation is a dominant strategy, and the myopic decision to buy is an optimal choice.  Even though the monopolist is facing strategic buyers, the monopolist's decision problem is reduced to a single person optimization problem.    An important extension of our exercise is to examine the duopoly market, where two firms complete through closely related products, say substitutes.    The decision problem of each duopolist interacts with each other.  The choice of the model of the other firm's behavior is determined as an equilibrium outcome.

\appendix
\footnotesize

\section{Proof of Proposition \ref{pr: lower bound}}
\label{Proof of Proposition pr: lower bound}

We first prove that $\dim D_t \ge 2$.   Since the market outcome in period $t$ is $(p_t,q_t)$, we show that if the algorithm uses only one of the two elements, the algorithm cannot PAC guarantee ${\mathcal F}$.   We only show that if $D_t=\{p_t\}$, the algorithm fails to PAC guarantee ${\mathcal F}$.    The proof for the case where $D_t=\{q_t\}$ follows the symmetric argument.

Let $F$ and $F^\alpha$  be the distribution functions of the uniform distribution over intervals $[{\underline p}, {\underline p}+1]$ and  $[{\underline p}+\alpha, {\underline p}+\alpha +1]$.   Note that the optimal price of $F^\alpha$ is
  \[
b^*(F^\alpha)=\frac{1+{\underline p}+\alpha}{2}.
  \]
  Suppose that $D_t=\{p_t\}$.   Then,
  \[
    {\mathcal A}({\mathcal D}_t: F^\alpha)={\mathcal A}({\mathcal D}_t: F)
    \qquad\forall\alpha.
  \]
Both  ${\mathcal A}({\mathcal D}_t: F^\alpha)$ and ${\mathcal A}({\mathcal D}_t: F)$ have the same input and, therefore, generate the identical forecast for any $F^\alpha$.
  \begin{equation}
    \varphi_p({\mathcal A}({\mathcal D}_t: F^\alpha))=\varphi_p({\mathcal A}({\mathcal D}_t: F))
    \qquad\forall\alpha, \forall t\ge 1.
    \label{eq: identical forecast}
  \end{equation}
  where $\varphi$ is the translating function of the forecast of ${\mathcal A}$ to the optimal price.
Since ${\mathcal A}$ PAC guarantees ${\mathcal F}$,
\[
  \lim_{t\rightarrow\infty}\varphi_p({\mathcal A}({\mathcal D}_t: F^\alpha))
  =b^*(F^\alpha)=\frac{1+{\underline p}+\alpha}{2} \qquad\forall\alpha
\]
with probability 1.
If so, \eqref{eq: identical forecast} cannot hold.

Next, we show that $\dim ({\mathcal A}({\mathcal D}_{t-1})) \ge 2$.
Suppose that $\dim ({\mathcal A}({\mathcal D}_{t-1})) =1$.     Let
\[
\theta_t={\mathcal A}({\mathcal D}_{t})\in {\mathbb R} 
\]
by the hypothesis of the proof.     Since $\Theta$ is compact, $\{\theta_t\}$ has a convergent subsequence.  After re-numbering the subsequence, let
\[
\theta_t\rightarrow \theta^*\in\Theta.
\]
Since ${\mathcal A}$ PAC guarantees ${\mathcal F}$, $\forall F\in {\mathcal F}$,
\[
\theta_t={\mathcal A}({\mathcal D}_t: F) \rightarrow \theta^*
\]
implies that
\[
\varphi(\theta_t)\rightarrow\varphi(\theta^*)=(b^*(F), q^*(F)).
\]
 We can therefore consider the following set: 

\begin{equation} 
  \{(a, b) \mid \exists \theta\in \Theta \text{ such that } a = \varphi_p(\theta), b =\varphi_q(\theta) \}
  \label{eq: image of forecast}
\end{equation}

$\theta\in {\mathbb R}$ is single-dimensional by the hypothesis of the proof.
Lipschitz continuity of $\varphi$ implies that the Hausdorff dimension of the set $\Theta$ cannot increase under the parameterization (see, for instance, Proposition 3.1.5 of \citeasnoun{AmbrosioTilliTextbook}). Hence this set is one dimensional.   

Since ${\mathcal A}$ PAC guarantees ${\mathcal F}$, \eqref{eq: image of forecast} is exactly
\begin{equation} 
\{(a, b) \mid \exists F\in{\mathcal F} \ \text{ such that } a = b^{*}(F), b =1-F(b^{*}(F)) \} \label{eq:contradicteq}
\end{equation}

To obtain the contradiction, it suffices to observe that \eqref{eq:contradicteq} is in fact not one dimensional. This can be seen using linear demand curves (or the uniform distribution function $F$) which is a subset of ${\mathcal F}$. That is, if $q=b_{0}+ \beta_{1} p$, then $q^{*} = \frac{\beta_{0}}{2}$ and $p^{*} = - \frac{ \beta_{0}}{2 \beta_{1}}$. Since for every fixed $\beta_{0}$, (a) the set of $(q^{*}, p^{*})$ generated by some $\beta_{1}$ is a one dimensional set, (b) these sets are also disjoint for distinct $\beta_{0}$, and (c) the set of valid $\beta_{0}$ is itself one dimensional, we therefore have that the set in (\ref{eq:contradicteq}) is two dimensional for this class of $\mathcal{F}$. This contradiction establishes that there cannot be a single-dimensional parameter which is used in the algorithm. 

\section{Proof of Theorem \ref{thm:mainresult}}
\label{Proof of Theorem thm:mainresult}

The projection facility is only used to ensure the tightness of the set of the sample paths. It does not alter the asymptotic properties such as the stability and the large deviation properties of the algorithm  (\citeasnoun{DupuisandKushner89}).     Thanks to the projection facility, we can assume that $(\beta_{0,t},\beta_{1,t})$ is contained in a compact convex set.  For the remainder of the paper, we suppress the projection facility to simplify the exposition when we examine the asymptotic properties of the algorithm. 

\subsection{Preliminaries}

Fix $\tau>0$ and consider an interval $[0,\tau)$ of real-time.    Fix small $a>0$, and divide the interval into subintervals of size $a$, with a possible exception of the last subinterval.    Define
\[
T(a)=\left\lceil \frac{\tau}{a} \right\rceil
\]
be the number of the subintervals (treating the last subintervals as the full size subinterval) in $[0,\tau)$, where $a$ is the gain coefficient of the updating term in the recursive formula. Recall that $\epsilon_{2,t}$ is distributed uniformly over $[-\epsilon,\epsilon]$ where $\epsilon>0$.
We will choose $\tau, a,\epsilon$ to meet the algorithm's accuracy and confidence requirement, which in turn determines $T(a)$.

For a fixed $F\in {\mathcal F}^\eta$, we are interested in
\[
\beta_{T(a)}-\beta^*(F).
\]
For $t\ge 1$, we can write the recursive formula as
\[
\beta_t =\beta_{t-1}+a \varphi(\beta_{t-1},p_t,\epsilon_t)
\]
since the updating term is determined by the old estimate $\beta_{t-1}$, the price in period $t$ and the realized quantity, where the last two variables are subject to two shocks $(\epsilon_{1,t},\epsilon_{2,t})$.
Let
\[
\varphi(\beta_{t-1},p_t,\epsilon_t) =\Expect_{t-1}\varphi(\beta_{t-1},p_t,\epsilon_t) +\xi_t
\]
where $\xi_t$ is the martingale difference.   Since $\beta_t\in {\sf B}$ which is compact, $\xi_t$ is uniformly bounded: $\exists\xi>0$ such that
\[
\abs{\xi_t}\le \xi.
\]
Define
\[
\bbar_{t-1}(\beta_{t-1})=\Expect_{t-1}\varphi(\beta_{t-1},p_t,\epsilon_t).
\]
As shown in \eqref{eq: simple calculation}, the functional form of $b_{t-1}$ is not affected by $t-1$ and is a Lipschitz continuous function of  $\beta_{t-1}$.   To simplify notation, we write
$\bbar(\beta_{t-1})$ in place of $\bbar_{t-1}(\beta_{t-1})$, dropping the time subscript from $\bbar_{t-1}$.   We can write the recursive formula as
\[
\beta_t =\beta_{t-1}+a \left[ \bbar(\beta_{t-1})+\xi_t \right].
\]
Given $\beta_0,\beta_1,\ldots,\beta_{T(a)}$, define a continuous time process obtained by the linear interpolation: $\forall s\in [a(t-1),at)$,
\[
\beta^a(s)=\frac{(s-a(t-1))\beta_t+(at-s)\beta_{t-1}}{a}.
\]
Define
\[
\beta(s)=\lim_{a\rightarrow 0}\beta^a(s)
\]
pointwise. The existence of the limit point is guaranteed by the fact that $\beta_t$ is contained in a compact set and $\bbar$ is a Lipschitz continuous function.

Define $\beta^*(F)=(\beta^*_0(F),\beta^*_1(F))$ as the intercept and the slope of a linear demand curve that generates the optimal price $b^*(F)$ and the expected quantity $1-F(b^*(F))$, that solves
\[
  1-F(b^*(F))=\frac{\beta^*_0(F)}{2} \ \ \text{and} \ \
  f(b^*(F))=-\beta_1^*(F).
\]
We can write
\begin{eqnarray}
\beta_{T(a)}-\beta^*(F) 
  & = & \beta_0 +a\sum_{t=1}^{T(a)} \bbar(\beta_{t-1}) +a\sum_{t=1}^{T(a)} \xi_t -\beta^*(F) \nonumber \\
  & = & \beta_0+\int_0^\tau \bbar(\beta(s))ds -\beta^*(F) \label{eq: ODE term}  \\
  && \ \ \ \ \
     +a\sum_{t=1}^{T(a)} \bbar(\beta_{t-1}) - \int_0^\tau \bbar(\beta(s))ds \label{eq: Riemann} \\
  && \ \ \ \ \
     +a\sum_{t=1}^{T(a)}\xi_t  \label{eq: noise} 
\end{eqnarray}

We examine \eqref{eq: ODE term}, \eqref{eq: Riemann} and \eqref{eq: noise} one by one.

\subsection{Convergence and Stability}

We can write
\[
\beta(\tau)=\beta(0)+\int_0^\tau \bbar(\beta(s))ds
\]
where $\beta(0)=\beta_0$, which is often written as
\[
{\dot\beta}=\bbar(\beta).
\]
To simplify notation, we write
\[
b_t=-\frac{\beta_{0,t-1}}{2\beta_{1,t-1}} \qquad\forall t\ge 1
\]
and $b(\tau)$ as the continuous process constructed from $b_t$ via linear interpolation $\forall \tau\ge 0$.    If the meaning is clear from the context, we drop $\tau$ to write $b$ instead of $b(\tau)$.   The same convention applies to all other variables such as $\beta_{0,t}$ and $\beta_{1,t}$.

We examine the properties of the ordinary differential equation (ODE):
\[
  {\dot \beta} = R^{-1}
\Expect  \left[
    \begin{matrix}
      1 \\
      -b_t+\epsilon_{1,t}
    \end{matrix}
  \right]
  \phi(\beta_{t-1},\epsilon_{1,t})  
\]
where
\[
  R=\left[
\begin{matrix}
  1  &    b_t  \\
b_t  &  b^2_t +\sigma^2_1 
\end{matrix}
    \right]
\]
and
\[
  \phi(\beta_{t-1},\epsilon_t)
  =1-F\left( b_t+\epsilon_{1,t}  \right)+\epsilon_{2,t} -
\beta_{0,t-1}-\beta_{1,t-1}b_t-\beta_{1,t-1}\epsilon_{1,t}.
\]
Since $\epsilon_{1,t}$ has small support, and $F$ is differentiable, it is more convenient to write
\[
  F\left( b_t+\epsilon_{1,t}  \right)=
  F\left( b_t\right)+f\left( b_t\right)\epsilon_{1,t} +\bigO(\epsilon^2).
\]
We are interested in the column vector
\[
\Expect \left[
    \begin{matrix}
      1 \\
      b_t+\epsilon_{1,t}
    \end{matrix}
  \right]
  \phi(\beta_{t-1},\epsilon_t).
\]
The first component is
\[
  1-  F\left( b_t\right) -\beta_{0,t-1}-\beta_{1,t-1}b_t
  +\bigO(\epsilon^2).
\]
The second component is
\[
\left[-f\left(\frac{\beta_{0,t-1}}{2\beta_{1,t-1}}\right) -\beta_{1,t-1}\right]\sigma^2_1+\bigO(\epsilon^3).
\]
We can write
\begin{equation}
  {\dot\beta}=R^{-1}
  \left[
\begin{matrix}
 1-  F\left( b\right) -\beta_{0}-\beta_{1}b
  +\bigO(\epsilon^2) \\
b\left[ 1-  F\left( b\right) -\beta_{0}-\beta_{1}b\right]-\left[f\left( b \right)+\beta_{1}\right]\sigma^2_1+\bigO(\epsilon^3)
\end{matrix}
  \right].
  \label{eq: simple calculation}
\end{equation}
At the stationary point $b^*(F)$ where the right hand side of ODE vanishes,
\begin{eqnarray*}
  1-F\left( b^*(F) \right) & = & \frac{\beta_0}{2} -\bigO(\epsilon^2) \\
  f\left( b^*(F) \right) & = & -\beta_1 +\frac{\bigO(\epsilon^3)}{\bigO(\epsilon^2)}.
\end{eqnarray*}
Thus, 
\begin{equation}
  \frac{ 1-F\left( b^*(F) \right)}{f\left( b^*(F) \right)} =b^*(F)+\bigO(\epsilon).
  \label{eq: stationary solution}
\end{equation}
To simplify notation, let us ignore $\bigO$ term and treat $b^*(F)$ as the solution of
\begin{equation*}
  \frac{ 1-F\left( b^*(F) \right)}{f\left( b^*(F) \right)} =b^*(F).
\end{equation*}
The uniqueness of the solution is implied by the increasing hazard rate property of $F$.
Since $\forall F\in {\mathcal F}$, $f$ is Lipschitz continuous 
\[
|f(x)-f(x')| \le \eta |x-x'|,
\]
$|\bigO(x)| \le \eta |x|$.
By reducing the size of the support of $\epsilon_{1,t}$, we can achieve the desired level of accuracy uniformly.

Let us proceed with the calculation after suppressing $\bigO$ terms.   Note that
\[
  R^{-1}=\frac{1}{\sigma^2_1}
\left[
\begin{matrix}
b^2 +\sigma^2_1   &    -b  \\
              -b  &  1
\end{matrix}
    \right]
\]
We write ODE of $\beta=(\beta_0,\beta_1)$ without linear approximation error $\bigO$ as
\[
  \left[ \begin{matrix}
  {\dot \beta}_0 \\
  {\dot \beta}_1
\end{matrix} \right]
  =
\left[
\begin{matrix}
(1-F(b)-\beta_0-\beta_1 b) +b (f(b)+\beta_1) \\
-(f(b)+\beta_1)
\end{matrix}
  \right].
\]
By definition of $b$
\[
\beta_0+2\beta_1 b=0
\]
at every moment of time.   Thus,
\[
{\dot\beta}_0+2{\dot\beta}_1 b +2\beta_1{\dot b}=0.
\]
After substituting ${\dot \beta}_0$ and ${\dot\beta}_1$, we have
\begin{equation}
{\dot b}=-\frac{f(b)}{2\beta_1}\left[ \frac{1-F(b)}{f(b)}-b \right]\equiv {\sf R}(b)
\label{eq: ODE two}
\end{equation}
modulo linear approximation errors $\bigO(\epsilon)$.

\begin{lemma}  \label{lm: c exists}
$\exists c >0$ such that
\[
{\sf R}(b) \ge -c (b-b^*(F)) \qquad b\le b^*(F)
\]
and
\[
{\sf R}(b)\le -c (b-b^*(F)) \qquad b\ge b^*(F).
\]
\end{lemma}

\begin{proof}
We constructed the projection facility so that $\beta_1<0$ and $\beta_0>0$ and moreover,
\[
\sup_{F\in {\mathcal F}^\eta}\beta_1 <0.
\]
If $\inf_{F\in {\mathcal F}^\eta}f(b)>0$, then the proof is trivial.   Since we only assume that $F\in {\mathcal F}^\eta$, we need more work.

Consider an iso-(expected) profit curve in the space of $(q,p)$
\[
\Pi = pq.
\]
Its slope is
\[
-\left.\frac{dq}{dp}\right|_{\Pi} =\frac{1-F(p)}{p}.
\]
If $b=b^*(F)$, then the slope of the iso-profit curve must be equal to the slope of the demand curve $f(p)$ at $p=b^*(F)$:
\[
\frac{1-F(b^*(F))}{b^*(F)} =f(b^*(F)).
\]
Since $b^*(F)\in [{\underline v}, {\overline v}]$, the slope of an iso-profit curve at the optimal price must be uniformly bounded: $\exists {\underline M}, {\overline M}>0$ such that
\[
{\underline M}\le f(b^*(F))=\frac{1-F(b^*(F))}{b^*(F)} \le {\overline M}.
\]
Since $F\in{\mathcal F}^\eta$ is uniformly Lipschitz continuous, 
for a sufficiently small $\epsilon>0$,
\[
f(b^*(F)-\epsilon)\ge {\underline M}-\eta\epsilon > \frac{{\underline M}}{2}
\]
and similarly,
\[
f(b^*(F)+\epsilon)\le -{\underline M}+\eta\epsilon < -\frac{{\underline M}}{2}.
\]
We prove that the right hand side of the ODE
\[
{\dot b}=-\frac{f(b)}{\beta_1} \left( \frac{1-F(b)}{f(b)} -b\right)
\]
is strictly bounded away from 0 over $b <b^*(F)-\epsilon$ and $b>b^*(F)+\epsilon$.

If $f(b)=0$, the increasing hazard rate property implies that $f(b')=0$ $\forall b'<b$, and in particular $f({\underline v})=0$.
Since $f(b^*(F))>0$, $b < b^*(F)$.
By the construction of the algorithm along the boundary,
${\dot b}=-\frac{1}{2\beta_1} >0$ uniformly, because $\sup_{F\in {\mathcal F}^\eta}\beta_1<0$.

Suppose $f(b)>0$ and $b\le b^*(F)-\epsilon$.
We know that ${\sf R}(b)$ defined in \eqref{eq: ODE two} is strictly decreasing because of the increasing hazard rate property.   We also know that ${\sf R}(b^*(F)-\epsilon) \ge \frac{{\underline M}}{2}$.   Thus,
\[
{\sf R}(b) \ge {\sf R}(b^*(F)-\epsilon) \ge \frac{{\underline M}}{2} >0.
\]
Similarly, if $b \ge b^*(F)+\epsilon$,
\[
{\sf R}(b) \le {\sf R}(b^*(F)+\epsilon) \le -\frac{{\underline M}}{2} <0.
\]
The increasing hazard rate property implies that  if $b> b^*(F)$,
\[
  \frac{1-F(b)}{f(b)}-b =
  \frac{1-F(b)}{f(b)}-b -\left(\frac{1-F(b^*(F))}{f(b^*(F))}-b^*(F)\right) < -(b-b^*(F))
\]
and if  $b< b^*(F)$,
\[
  \frac{1-F(b)}{f(b)}-b > -(b-b^*(F)).
\]
Thus,
\[
| {\sf R}(b)-{\sf R}(b^*(F)) | \ge  \frac{{\underline M}}{2} | b-b^*(F) |
\]
over $b\in [b^*(F)-\epsilon,b^*(F)+\epsilon]$. Since $b\in [{\underline v},{\overline v}]$, $\exists c >0$ such that
\[
{\sf R}(b) -{\sf R}(b^*(F))=R(b) \ge -c (b-b^*(F)) \qquad b\le b^*(F)
\]
and
\[
{\sf R}(b) -{\sf R}(b^*(F))=R(b) \le -c (b-b^*(F)) \qquad b\ge b^*(F).
\]
\end{proof}

For any initial value $b(0)\in [\vlbar,\vubar]$,
\[
\left|  b(\tau) -b^*(F)\right| \le e^{-c\tau} \left| b(0)-b^*(F)\right|\le e^{-c\tau}(\vubar-\vlbar).
\]
Let $\tau(\mu)$ be the first time when 
\[
\left| b(\tau) -b^*(F) \right| \le \mu.
\]
$(1-F(b^*(F)),b^*(F))\in K$ $\forall F\in {\mathcal F}$ and $K$ is a compact subset in the interior of ${\mathbb R}^2_+$.  Thus,
\[
\taubar (\mu) =\sup_{(\beta_0(0),\beta_1(0))\in {\sf B}}\tau(\mu) <\infty.
\]
and
\[
\taubar(\mu) \sim -\log\mu
\]
as $\mu\rightarrow 0$.   Let us choose $\tau=\taubar(\mu)$.

\subsection{Riemann Residual}

Let us consider \eqref{eq: Riemann}
\[
{\bf R}(a, F)=a\sum_{t=1}^{T(a)} \bbar(\beta_{t-1}) - \int_0^\tau \bbar(\beta(s))ds
\]
which is the Riemann residual.   Since $f$ is uniformly Lipschitz over ${\mathcal F}$, $\bbar(\beta)$ is uniformly Lipschitz:  $\exists\eta'>0$ such that
\[
\left| \bbar(\beta)-\bbar(\beta') \right|\le \eta' | \beta -\beta'| \qquad \forall F\in {\mathcal F}.
\]
For each subinterval of size $a$, the difference between the discrete value and the integration is at most $\frac{\eta' a^2}{2}$.   Thus,
\[
  {\bf R}(a, F) \le \frac{\eta' a^2}{2} \frac{\taubar(\mu)}{a}
  =\frac{\eta' a \taubar(\mu)}{2}.
\]
Note that the right hand side is independent of $F$.   Thus, $\forall\mu>0$,
define
\begin{equation}
{\overline a}=\frac{2\mu}{\taubar(\mu)\eta'}
\label{eq: aone}
\end{equation}
so that $\forall a\le {\overline a}$, $\forall F\in {\mathcal F}$,
\begin{equation}
{\bf R}(a,F) \le \mu.
\label{eq: Riemann upper bound}
\end{equation}
Thus,
\[
{\overline a}=\bigO\left( -\frac{\mu}{\log\mu}\right)
\]
which implies
\[
\frac{\tau(\mu)}{{\overline a}}=\bigO\left(\frac{(\log\mu)^2}{\mu} \right).
\]

\subsection{Lower Bound of Confidence}

Next, we examine \eqref{eq: noise}.   Consider $\frac{\xi_t}{\xi}$, which is a martingale difference with
  \[
    \left| \frac{\xi_t}{\xi}\right| \le 1 \ \ \text{and} \ \ 
    \Expect \left(\frac{\xi_t}{\xi}\right)^2\le 1.
  \]
Since $\frac{\xi_t}{\xi}$ satisfies the large deviation property, $\forall\mu'>0$, $\exists\rho(\mu',F)>0$ (called the rate function) such that
\[
\Prob \left( 
\left| \frac{1}{T}\sum_{t=1}^{T}  \frac{\xi_t}{\xi} \right| >\mu'
\right) \le e^{-\rho(\mu',F) T}.
\]
We need a uniform rate function $\rho(\mu',F) >0$ over ${\mathcal F}$.   By Azuma-Hoeffding-Bennett inequality (Corollary 2.4.7 in \citeasnoun{DemboandZeitouni98}), we have $\forall \mu' \in (0, 1/2)$,
\[
\Prob\left(
\left|\frac{1}{T}\sum_{t=1}^{T} \frac{\xi_t}{\xi}\right| > \mu'
\right)
\le e^{-2T {\sf H}\left(\mu' +\frac{1}{2}\mid \frac{1}{2}\right)}
\]
where
\[
{\sf H}(p \mid p_0)=p \log\frac{p}{p_0}+(1-p)\log\frac{1-p}{1-p_0}
\]
for $p,p_0\in (0,1)$.
Let
\[
T(a,\mu)=\left\lceil \frac{\taubar(\mu)}{a}\right\rceil.
\]
Then,
\[
\Prob\left(  \left| a \sum_{t=1}^{T(a,\mu)} \xi_t \right| <\taubar(\mu)\xi\mu' \right)
\le e^{-2T{\sf H} \left(\mu' +\frac{1}{2}\mid \frac{1}{2}\right)}.
\]
Let
\[
  \mu'=\frac{\mu}{\taubar(\mu)\xi}.
\]
After substitution, we have
\[
\Prob\left( \left| \frac{1}{T(a,\mu)} \sum_{t=1}^{T(a,\mu)} \xi_t \right| >\mu \right)
\le e^{-2 T(a,\mu){\sf H} \left(\mu' +\frac{1}{2}\mid \frac{1}{2}\right)}.
\]
Since $\taubar(\mu)=\bigO(-\log \mu)$,
\[
\frac{\mu}{\taubar(\mu)\xi} =\bigO(\mu)
\]
and therefore, $T(a,\mu)$ increases at the polynomial speed as $\mu\rightarrow 0$.
Let
\[
\rho =2{\sf H} \left(\mu' +\frac{1}{2}\mid \frac{1}{2}\right)>0.
\]

\subsection{Combine the Pieces}

Fix $\mu>0$.  Recall that we assume that $\epsilon_{1,t}$ is distributed over $[-\epsilon,\epsilon]$.  We first choose $\epsilon>0$ so that the stationary solution \eqref{eq: stationary solution} of ODE is within $\mu$ neighborhood of $\beta^*(F)$ $\forall F\in {\mathcal F}$.

Since $\forall f\in {\mathcal F}$
\[
|f(p)-f(p')| \le \eta |p- p'|.
\]
The Taylor residual is bounded by $\eta \epsilon$ for small $\epsilon>0$:
\[
|\bigO(\epsilon^3)| \le |\bigO(\epsilon^2)| \le \eta\epsilon^2 \qquad\forall t\ge 1. 
\]
Note that the last term is independent of $F\in {\mathcal F}$. Choose $\epsilon>0$ sufficiently small so that
\[
\left| \beta^s - \beta^*(F) \right| <\mu \qquad\forall F\in {\mathcal F}
\]
where $\beta^s$ solves \eqref{eq: stationary solution}.  We chose $\taubar(\mu)$ so that
\[
| \beta(\taubar(\mu)) -\beta^s | <\mu.
\]
Given $\mu>0$, we chose $a_1>0$ in \eqref{eq: aone}  so that the Riemann residual ${\sf R}(a,F)$ satisfies \eqref{eq: Riemann upper bound}.

By the construction, $\forall a\in (0,a_1)$, 
\[
\Prob\left( \left| \frac{1}{T(a,\tau}\sum_{t=1}^{T(a,\tau)} \xi_t \right| >\mu \right) < e^{-T(a,\tau)\rho}.
\]
Combining these results, we have $\forall a\in (0,a_1)$, 
\[
\Prob \left( \left| \beta_{T(a,\tau)}-\beta^*(F) \right| \ge 4\mu\right) \le e^{-T(a,\tau)\rho}
\]
where $T(a,\tau)$ increases linearly with respect to $1/a$ and at the polynomial speed with respect to $1/\mu$.

\newpage
\normalsize
\bibliographystyle{econometrica}
\bibliography{adaboost}

\end{document}